\documentclass[sigconf, authorversion]{acmart}
\usepackage[utf8]{inputenc}
\usepackage{multirow}
\usepackage{graphicx}
\usepackage{textcomp}
\usepackage{xcolor}
\usepackage{xspace}
\usepackage{amsthm}
\usepackage{subfig}
\usepackage{booktabs}
\usepackage{graphics}
\usepackage{mathtools}
\usepackage{hyperref}
\usepackage[font=small,labelfont=bf]{caption}
\usepackage{algorithm}
\usepackage[noend]{algpseudocode}

\definecolor{mycolor}{rgb}{0,0.5,0.8}
\definecolor{mycolor2}{rgb}{0.9,0.0,0.0}
\definecolor{mygreen}{rgb}{0.0,0.55,0.0}
\definecolor{mynewgreen}{rgb}{0.76,0.7,0.5}
\definecolor{myred}{rgb}{0.55,0.0,0.0}
\definecolor{mygrey}{rgb}{0.52,0.52,0.51}
\definecolor{plotcolor}{rgb}{0.95,0.66,0.34}
\definecolor{textcolor1}{rgb}{0.0, 0.43, 0.64}

\DeclarePairedDelimiter\ceil{\lceil}{\rceil}
\DeclarePairedDelimiter\floor{\lfloor}{\rfloor}

\newtheoremstyle{sig}
  {}
  {}
  {\itshape}
  {}
  {\scshape}
  {.}
  {.5em}
  {#1 #2\thmnote{\quad(#3)}}

\theoremstyle{sig}

\newtheorem{theorem}{Theorem}
\newtheorem{problem}{Problem}

\newtheorem{dfn}{Definition}
\newtheorem{fact}{Fact}

\newtheorem{lemma}[theorem]{Lemma}

\newtheorem{example}{Example}

\newtheorem{prop}{Proposition}[dfn]

\copyrightyear{2022} 
\acmYear{2022} 
\setcopyright{acmcopyright}\acmConference[WWW '22]{Proceedings of the ACM Web Conference 2022}{April 25--29, 2022}{Virtual Event, Lyon, France}
\acmBooktitle{Proceedings of the ACM Web Conference 2022 (WWW '22), April 25--29, 2022, Virtual Event, Lyon, France}
\acmPrice{15.00}
\acmDOI{10.1145/3485447.3512205}
\acmISBN{978-1-4503-9096-5/22/04}

\begin{document}
\setlength{\textfloatsep}{1.5pt}
\title{FirmCore Decomposition of Multilayer Networks}

\author{Farnoosh Hashemi}  
\authornote{These authors contributed equally.}
\affiliation{%
  \institution{University of British Columbia}
  \city{Vancouver}
  \state{BC}
  \country{Canada}
}
\email{farsh@cs.ubc.ca}

\author{Ali Behrouz}
\authornotemark[1]
\affiliation{%
  \institution{University of British Columbia}
  \city{Vancouver}
  \state{BC}
  \country{Canada}
}
\email{alibez@cs.ubc.ca}

\author{Laks V.S. Lakshmanan}
\affiliation{
  \institution{University of British Columbia}
  \city{Vancouver}
  \state{BC}
  \country{Canada}
}
\email{laks@cs.ubc.ca}

\ccsdesc[500]{Mathematics of computing~Graph algorithms}

\begin{abstract}
  \noindent 
  A key graph mining primitive is extracting dense structures from graphs, and this has led to interesting notions such as $k$-cores which subsequently have been employed as building blocks for capturing the structure of complex networks and for designing efficient approximation  algorithms for challenging problems such as finding the densest subgraph. In applications such as biological, social, and transportation networks, interactions between objects span multiple aspects. 
  Multilayer (ML) networks have been proposed for accurately modeling such applications. In this paper, we present \textit{FirmCore}, a new family of dense subgraphs in ML  networks, and show that it satisfies many of the nice properties of $k$-cores in single-layer graphs. Unlike the state of the art core decomposition of ML graphs, FirmCores have a polynomial time algorithm, making them a powerful tool for understanding the structure of massive ML networks. We also extend FirmCore for directed ML graphs. We show that FirmCores and directed FirmCores can be used to obtain efficient approximation algorithms for finding the densest subgraphs of ML graphs and their directed counterparts. Our extensive experiments over several real ML graphs show that our FirmCore decomposition algorithm is significantly more efficient than known algorithms for core decompositions of ML graphs. Furthermore, it returns solutions of matching or better quality for the densest subgraph problem over (possibly directed) ML graphs.  
\end{abstract}

\keywords{
Graph mining, $k$-core, multi-layer graph, densest subgraph.}

\maketitle
\vspace{0ex}
\section{Introduction}
\label{sec:Introduction}
In applications featuring complex networks such as social, biological, and transportation networks, the interactions between objects tend to span multiple aspects. E.g., interactions between people can be social or professional, and  professional interactions can differ according to topics. Accurate modeling of such applications has led to multilayer (ML) networks~\cite{main-ML}, where nodes can have interactions in multiple layers but there are no inter-layer edges. They have since gained popularity in an array of applications in social and biological networks and in opinion dynamics \cite{social-ml, ML-Covid, ML-IM, Ml-bio}.

\begin{example}
Figure \ref{fig:example}(a) is a  network 
showing three groups of researchers collaborating in various topics, and Figure \ref{fig:example}(b) is its corresponding multilayer perspective, where each layer represents collaborations in an individual topic. While all three groups have the same structure in the single-layer model, the ML graph model  distinguishes these group structures, extracting their  complex relationships. 
\end{example}

\begin{figure}
    \centering
    \subfloat[][Single-layer perspective]{
        \begin{minipage}[c][1\width]{0.25\textwidth}%
        \begin{center}
        \includegraphics[width=0.75\textwidth]{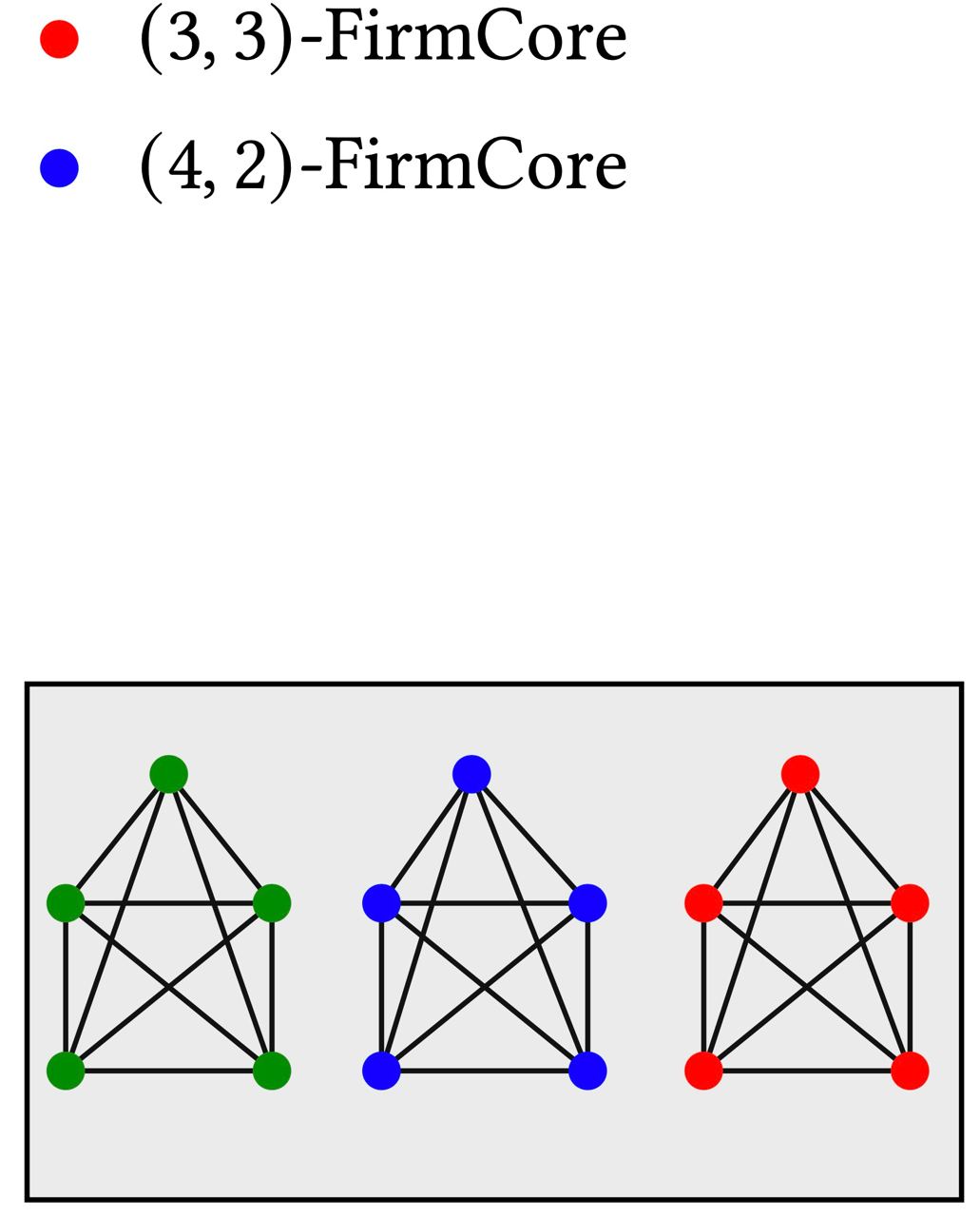}
        \end{center}
        \end{minipage}
    }
    \centering
    \subfloat[][Multilayer perspective]{
        \begin{minipage}[c][1\width]{0.25\textwidth}%
        \includegraphics[width=0.95\textwidth]{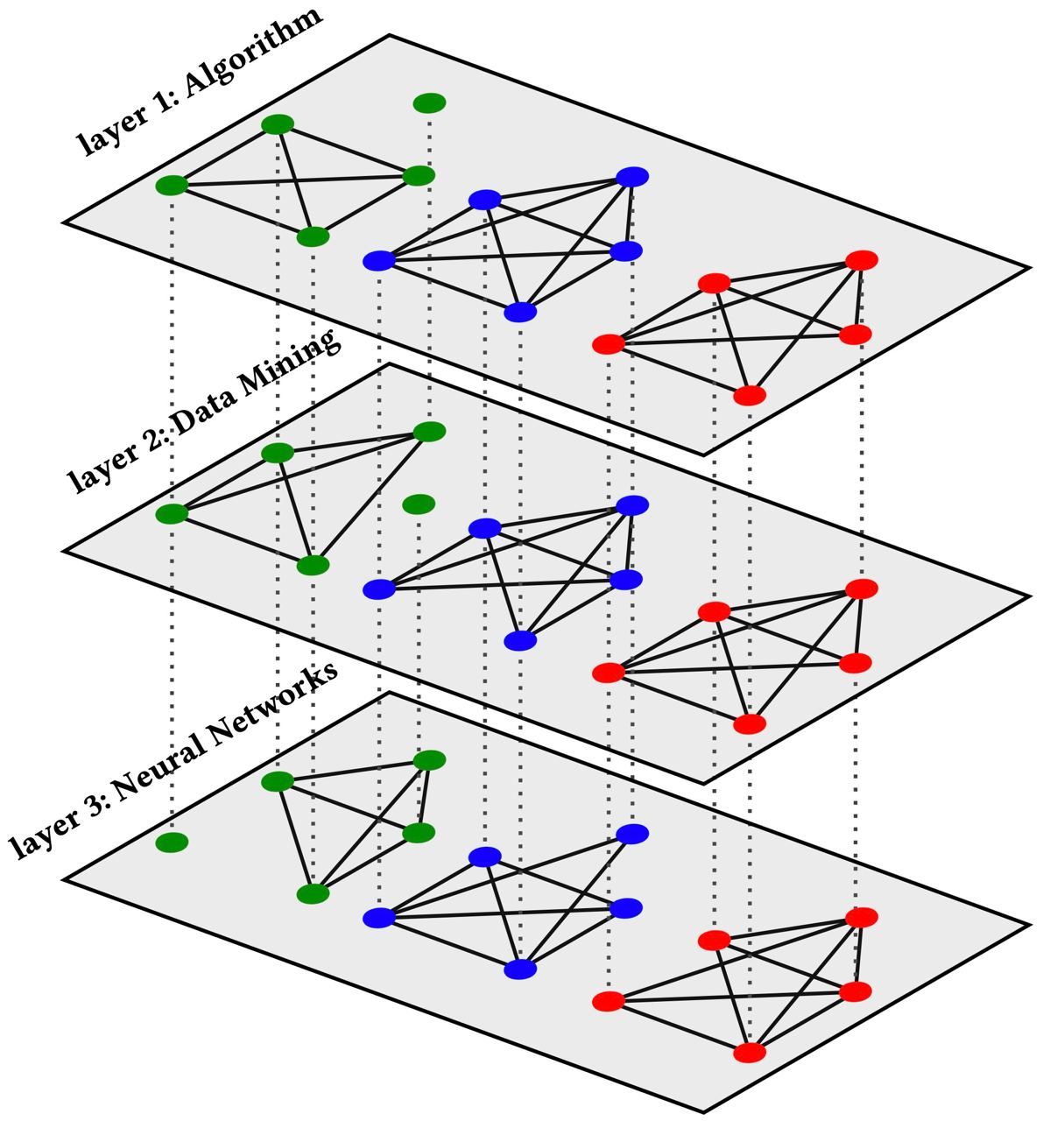}
        \end{minipage}
    }
    \vspace{-2ex}
    \caption{An example of a collaboration network.}
    \label{fig:example}
\end{figure}

Extraction of dense structures in a graph has become a key graph mining primitive with a wide range of applications \cite{bio-dense, finance-dense, web-dense}. 
The core decomposition  \cite{k-core-first} is a widely used  operation since it can be computed in linear time \cite{k-core-algorithm}, is linked to a number of different definitions of a dense subgraph \cite{speed-up-core}, and can be used for efficient approximations of quasi-cliques, betweenness centrality, and densest subgraphs. Azimi et al. \cite{azimi-etal} extended the classic notion of core for ML networks: given $L$ layers and a vector of $|L|$ dimensions \textbf{k}$=[k_\ell]_{\ell \in  L}$ with $k_\ell \in \mathbb{Z}^{\geq 0}$,  the \textbf{k}-core is a maximal subgraph, in which each node in layer $\ell$ has at least $k_\ell$ neighbors. Zhu et al. \cite{Coherehnt-core-ML} define the core of ML networks, termed $d$-core, as a maximal subgraph in which every node in \textit{each} layer has $\geq d$ neighbors in the subgraph. The core decomposition algorithms~\cite{MLcore, CoreCube} based on these two models have an \textit{exponential running time complexity in the number of layers}, making them  prohibitive for large graphs even with a small number of layers. 
Even for the special case of two layers, the state-of-the-art techniques for core mining of ML graphs  have quadratic \cite{Coherehnt-core-ML} and cubic \cite{MLcore} time complexity, making~them~impractical~for~large~networks. 

To~mitigate~the~above~limitations,~we~propose~the~notion~of~\textit{$(k, \lambda)$-FirmCore} as a maximal subgraph in which every node is connected to at least $k$ other nodes within that subgraph, in each of at least $\lambda$  individual layers. As we show, FirmCores retain the~hierarchical~structure of classic cores (for single-layer graphs) and this~leads~to~a \textit{linear time core decomposition of ML networks  in number of nodes/edges}.

In addition to its computational advantages over previous core definitions in ML networks \cite{azimi-etal, Coherehnt-core-ML}, FirmCore can also address some of their limitations and capture complex structures that may not be found by previous formulations. A major limitation of these models is forcing nodes to satisfy degree constraints in \textit{all} layers, including noisy/insignificant  layers \cite{MLcore}. 
These noisy/insignificant layers may be different for each node. Therefore, this hard constraint would result in missing some dense structures. In contrast, FirmCore uses the intuition that a dense structure in ML networks can be identified by a subgraph where each node has enough connections to others in a sufficient number of interactions types (aka layers). Accordingly, even if a node is not densely connected to others in some individual layers, it still can be considered as part of a dense structure if it has enough connections to others in sufficiently many other layers.
\vspace{-2ex}
\begin{example}
In Figure~\ref{fig:example}(b), in each layer, there is an isolated green node, researcher, who does not collaborate in that specific area. Thus, e.g., the top layer (``Algorithm'') is an insignificant layer for the isolated green node in this layer. While green nodes are a component of a $(3, 2)$-FirmCore, as per the lens of  \textbf{k}-core (resp. $d$-core), the only \textbf{k}-core (resp. $d$-core) that contains all green nodes is the $(0, 0, 0)$-core (resp. $0$-core), which includes the entire graph as well! Accordingly, these models cannot distinguish the green nodes from other sparse structures that may exist~in~a~network.  
\end{example}
\vspace{-2ex}

As a real-world motivation, users may have accounts in different social media. We can model this situation as an ML network whose layers  correspond to different~social~media.~Different~interests/themes~may bring together different groups of users in different layers. For instance, considering LinkedIn, FaceBook,  and Instagram, people may have different kinds of intra-layer connections. In general, a diverse community across different networks can be identified as a group in which every user has enough connections to others within the group in a sufficient~number~of~layers.


Directed graphs naturally arise in several applications, and their ML versions lead to directed ML networks. For instance, the Web can be modeled as an ML  directed graph in which each node represents a Web page, and each edge a hyperlink between the  corresponding pages, which is associated with a label. Depending on the application, this label can be a timestamp or the topic of the source page's citing text that points  to the target page. FirmD-Cores can be viewed as a multilayer generalization of communities based on hubs and authorities in directed graphs a la  \cite{gibson-etal-1998}: in a FirmD-Core $(S,T)$, vertices in $T$ can be viewed as akin to ``authorities" over a set of criteria (e.g., topics, time points, etc.) modeled as layers of an ML graph; sufficiently many vertices in $S$, akin to ``hubs", point to them in sufficiently many layers; similarly, sufficiently many authorities in $T$ are pointed to by the hub vertices in $S$ in sufficiently many layers.
\textit{To our knowledge, core decomposition of directed ML networks has not been studied before.}

We make the following contributions: \textbf{(1)} We introduce a novel dense subgraph model for ML networks, \textit{FirmCore} ($\S$~\ref{sec:FirmCore}). \textbf{(2)} We extend FirmCore to directed ML networks, FirmD-Core  ($\S$~\ref{sec:FirmDCore}). \textbf{(3)} We show that FirmCores and FirmD-Cores retain the nice properties of uniqueness and hierarchical containment. Leveraging this, we propose efficient algorithms for FirmCore and FirmD-Core decompositions ($\S$~\ref{sec:Algorithms}). \textbf{(4)} We propose the \textit{first polynomial-time approximation algorithm} for the densest-subgraph problem in undirected ML graphs while the current state-of-the-art has an  exponential runtime ($\S$~\ref{sec:MDS-algorithm}). \textbf{(5)} We extend the densest-subgraph problem to directed ML networks and propose an efficient approximation algorithm based on our FirmD-Core decomposition ($\S$~\ref{sec:MDS-algorithm}). \textbf{(6)} Based on comprehensive experiments on sixteen real datasets, we show that our approximation algorithm for densest-subgraph in undirected ML graphs is \textit{two orders of magnitude faster than the current state-of-the-art algorithm and has less memory usage},  while achieving similar or better quality compared to the state of the art  ($\S$~\ref{sec:Experiments}). For lack of space, we suppress some of the proofs. The omitted proofs, additional experiment results, as well as implementations and datasets can be found in the appendix.

\vspace{-2ex} 
\subsection{Related Work and Background}
\label{sec:RW} 

\noindent
\textbf{$k$-Core. } 
The $k$-core of a graph is the maximal subgraph where each node has degree $\geq k$ within the subgraph. Motivated by the linear-time core decomposition of graphs, $k$-cores have  attracted enormous attention, and have been employed in various applications such as analyzing social networks~\cite{core-social-network1}, anomaly detection~\cite{anomaly-detection1}, detection of influential spreaders~\cite{influence1}, and biology~\cite{drug-core}. Moreover, $k$-core has been extended to various types of graphs~\cite{uncertain-core, span-core, D-core}.

\noindent
\textbf{Densest Subgraph Discovery.} 
There are various versions of the problem of extracting a single-layer dense subgraph, e.g., \cite{clique-based-density, generalized-dense, dense-small-large}, but in view of the hardness of the problem, extracting dense subgraphs based on the average degree, being solvable in polynomial time,  has received much attention. Goldberg \cite{densest_first} proposes an 
algorithm based on maximum flow. To improve efficiency, Asahiro et al.~\cite{greedy-densest} propose a linear-time greedy-based $\frac{1}{2}$-approximation algorithm~\cite{appox-densest}. Kannan and Vinay~\cite{DirectedDensity} proposed the directed densest subgraph problem, and Khuller and Saha~\cite{exact-directed} developed an exact algorithm for this problem. While several approximation algorithms have been suggested~\cite{appox-densest, DirectedDensity}, Ma~et~al.~\cite{xy-core} proposed the most efficient approximation algorithms based on the notion of $[x, y]$-core. 

\noindent
\textbf{Dense Structures in Multilayer Networks. } 
Dense subgraph mining over ML networks is relatively less explored. Jethava et al. \cite{densest-common-subgraph} formulate the densest common subgraph problem and develop a linear-programming formulation. Azimi et al. \cite{azimi-etal} propose a new definition of core, \textbf{k}-core, over ML graphs. Galimberti et al. \cite{MLcore} propose algorithms to find all possible \textbf{k}-cores, and define the densest subgraph problem in ML networks as a trade-off between high density and number of layers exhibiting the high density, and propose a core-based $\frac{1}{2|L|^{\beta}}$-approximation algorithm; their algorithm takes exponential time in number of layers, rendering it impractical for large networks. Liu et al. \cite{CoreCube} propose the CoreCube problem for computing ML $d$-core decomposition on all subsets of layers. This problem is a special case of \textbf{k}-core decomposition~\cite{MLcore}, restricted to vectors where a subset of elements are set to $d$ and other elements set to $0$. Zhu et al. \cite{Coherehnt-core-ML} study the problem of diversified coherent $d$-core search. Jiang et al. \cite{quasi-clique-ml} study a related problem of extracting frequent cross-graph quasi-cliques. Finally, Wu et al. \cite{core-temporal} extend the notion of core~to~temporal graphs. While their definition can be adapted to ML networks, it is equivalent to collapsing the layers, removing edges that occur less than a threshold, and finding cores in the resulting  single-layer graph. It focuses on frequency of interactions in a manner appropriate for temporal graphs, however it cannot capture complex relationships in ML~networks.

Fang et al.~\cite{HIN-core} define the core of heterogeneous information networks. Liu et al.~\cite{core-bipartite} define the core of bipartite graphs and Zhou et al.~\cite{HIN2} extend it to $k$-partite networks. However, all these models emphasize heterogeneous types of entities connected by different relationships, which is different from the concept of ML networks. As such, their approach is inapplicable for ML networks.

\section{Problem Statement}
\label{sec:ProblemDefinition}
We let $G = (V, E, L)$ denote an undirected ML graph, where $V$ is the set of nodes, $L$ is the set of layers, and $E \subseteq V \times V \times L$ is the set of edges. 
The set of neighbors of node $v \in V$ in layer $\ell \in L$ is denoted $N_\ell(v)$ and the degree of $v$ in layer $\ell$ is $\text{deg}_\ell (v) = |N_\ell(v)|$. For a set of nodes $H \subseteq V$, $G_\ell[H] = (H, E_\ell[H])$ denotes the subgraph of $G$ induced by $H$ in layer $\ell$, and $\text{deg}^H_{\ell}(v)$ denotes  the degree of $v$ in this subgraph. We abuse  notation and denote  $G_\ell[V]$ and $E_\ell[V]$ as $G_\ell$ and $E_\ell$, respectively. 

\vspace*{-1ex} 
\subsection{FirmCore Decomposition}
\label{sec:FirmCore}
\begin{dfn}[FirmCore]\label{FirmCore}
Given an undirected ML graph $G$, an integer threshold $1 \leq  \lambda \leq |L|$, and an integer $k \geq 0$, the $(k, \lambda)$-FirmCore of $G$ ($(k, \lambda)$-FC for short) is a maximal subgraph $H = G[C_{k}] = (C_{k}, E[C_{k}], L)$ such that for each node $v\in C_k$ there are at least $\lambda$ layers $\{\ell_1, ..., \ell_\lambda\} \subseteq L$ such that $\text{deg}^H_{\ell_i}(v) \geq k$, $1\leq i\leq \lambda$.
\end{dfn}

Given $\lambda$, the \textit{FirmCore index} of $v$, $core_{\lambda}(v)$, 
is defined as the highest order $k$ of a $(k, \lambda)$-FC containing $v$. Next, we formally define the FirmCore decomposition problem and show that not only it is unique, but also it has the nested property. 
\vspace{-1mm}

\begin{problem}[FirmCore Decomposition]
Given an ML graph $G$, find the FirmCore decomposition of $G$, that is the set of all $(k,\lambda)$-FCs of $G$,  $k \geq 0$ and $1 \leq \lambda \leq |L|$.
\end{problem}

\begin{prop}[Uniqueness]
The FirmCore decomposition of~$G$ is unique. 
\label{prop:FirmCore1}
\end{prop}

\begin{prop}[Hierarchical Structure]\label{prop:FirmCore2}
Given a threshold $\lambda \in \mathbb{N}^+$, and an integer $k \geq 0$, the $(k + 1, \lambda)$-FC and $(k, \lambda + 1)$-FC of~$G$ are subgraphs of its $(k, \lambda)$-FC.
\end{prop}

\begin{example}
In Figure~\ref{fig:example}, red nodes are $(3, 3)$-FC as all red nodes~in~three layers have at least three red neighbors. Similarly, blue nodes are $(4, 2)$-FC as each node in at least two layers has~at~least~four~blue~neighbors.
\end{example}

Since a $k$-clique is contained in a $(k - 1)$-core, one application of single-layer core is speeding up the algorithm for finding maximum cliques, a classic NP-hard problem. Jiang et al. \cite{quasi-clique-ml} extend this problem and propose an algorithm for ML graphs. Specifically, given an ML graph $G = (V, E, L)$, a function $\Gamma~:~L~\rightarrow~(0,~1]$, a real threshold $min\_sup \in (0, 1]$, and an integer $min\_size \geq 1$, the \textit{frequent cross-graph quasi-clique} problem is to find all maximal subgraphs $G[H]$ of size $\geq min\_size$ such that $\exists$ at least $min\_sup\cdot|L|$ layers  $\ell \in L$ for which $G[H]$ is a $\Gamma(\ell)$-quasi-clique in layer $\ell$. The following result shows that FirmCore can  be used to speed up frequent cross-graph quasi-clique extraction. 

\begin{prop}
\label{prop:quasi-clique-prop}
Given an ML graph $G = (V, E, L)$, a function $\Gamma : L \rightarrow (0, 1]$, and a real threshold $min\_sup \in (0, 1]$, a frequent cross graph quasi-clique of size $k$ of $G$ is contained in the $(\gamma (k - 1), \lambda)$-FC where $\gamma = \min_{\ell \in L} \Gamma(\ell)$ and $\lambda = min\_sup\cdot|L|$.
\end{prop}

\subsection{FirmCore~Decomposition~of~Directed~Graph}
\label{sec:FirmDCore}
Let $G = (V, E, L)$ be a directed ML graph, and $\ell\in L$.  We use $\text{deg}^{-} (v)$ (resp. $\text{deg}^{-}_\ell (v)$) and $\text{deg}^{+} (v)$ (resp. $\text{deg}_\ell^{+} (v)$) to denote indegree (resp. indegree in layer $\ell$) and outdegree of $v$ (resp. outdegree in layer~$\ell$). 
 Given two sets $S, T \subseteq V$, which are not necessarily disjoint, we denote the set of all edges from $S$ to $T$ by $E(S, T)$. The subgraph induced by $S$ and $T$ with the edge set of $E(S, T)$ is called an $(S, T)$-induced subgraph of $G$ and is denoted $G[S, T] := (S \cup T, E[S, T], L)$. An induced subgraph $G[S, T]$ is maximal w.r.t. a property if there is no other induced subgraph $G[S', T']$ with $S \subseteq S'$ and $T \subseteq T'$, which satisfies that property. Now, we formally define the FirmD-Core in directed graphs.

\begin{dfn}[FirmD-Core]\label{FirmDCore}
Given a directed ML graph $G = (V, E, L)$, an integer threshold $1 \leq \lambda \leq |L|$, and integers $k, r \geq 0$, the $(k, r, \lambda)$-FirmD-Core of $G$ ($(k, r, \lambda)$-FDC for short) is a maximal $(S, T)$-induced subgraph $H = G[S, T] = (S \cup T, E[S, T], L)$ such that: 

    (1) $\forall v\in S$, there are at least $\lambda$ layers~$\ell\in L$~such that~$\text{deg}^{H,+}_{\ell}(v)~\geq~k$,
    
    (2) $\forall u\in T$, there are at least $\lambda$ layers  $\ell \in L$ such that $\text{deg}^{H,-}_{\ell}(u)\geq~r$.

\end{dfn}

Next we define the FirmD-Core decomposition problem and provide two key properties of it.

\begin{problem}[FirmD-Core Decomposition]
Given a directed ML graph $G = (V, E, L)$, find the FirmD-Core decomposition of $G$, that is the set of all $(k, r, \lambda)$-FDCs of $G$,  $k,r \geq 0$ and $1 \leq \lambda \leq |L|$.
\end{problem}


\begin{prop}[Uniqueness]
The FirmD-Core decomposition of~$G$ is unique. 
\label{prop:FirmDCore-prop}
\end{prop}

\begin{prop}[Hierarchical Structure]
Given a threshold $\lambda \in \mathbb{N}^+$, and two integers $k, r \geq 0$, the $(k + 1, r, \lambda)$-FDC, $(k, r+1, \lambda)$-FDC, and $(k, r, \lambda + 1)$-FDC of~$G$ are subgraphs of its $(k, r, \lambda)$-FDC.
\end{prop}

\vspace{-1ex}
\section{Algorithms}
\label{sec:Algorithms}
In this section, we design two exact algorithms for FirmCore and FirmD-Core decompositions.
\vspace{-1ex}
\subsection{FirmCore Decomposition}
\label{sec:alg_FirmCore}
One of the important challenges of dense structure mining in ML networks is time complexity. Since most complex networks, including social and biological networks, are massive, even non-linear algorithms can be infeasible for analyzing them. One potential fast algorithm to address this issue is to borrow the idea of $k$-core decomposition and iteratively remove a node that in at most $\lambda - 1$ layers has at least $k$ neighbors. However, in  contrast to single-layer networks, updating and checking a degree constraint on each node in ML graphs is not a simple task, since the degree of each node $v$ is defined as a vector $\text{deg}(v)$ with $|L|$ elements, $\ell$-th element being the degree of $v$ in layer $\ell$. 

The advantage of FirmCore is that we only need to look at the $\lambda$-th largest value in the degree vector, which we call  Top$-\lambda($deg$(v))$. 
\begin{fact}\label{fact:1}
For each $v \in V$, Top$-\lambda($deg$(v)) \geq core_\lambda(v).$
\end{fact}
For a given vertex $v$, if Top$-\lambda($deg$(v)) = k$, then it cannot be a part of $(k', \lambda)$-FirmCore, for $k' > k$. Therefore, given $\lambda$, we can consider Top$-\lambda($deg$(v))$ as an upper bound on the core index of $v$. In the FirmCore decomposition, we use  bin-sort \cite{k-core-algorithm}. Simply, we recursively pick a vertex $v$ with the lowest Top$-\lambda($deg$(v))$, assign its core number as its Top$-\lambda($deg$(v))$, and then remove it from~the~graph. 

\begin{algorithm}[t]
    \small
    \caption{$(k, \lambda)$-FirmCore}
    \label{alg:FirmCore}
    \begin{algorithmic}[1]
        \Require{An ML graph $G = (V, E, L)$, and a threshold $\lambda \in \mathbb{N}^+$}
        \Ensure{FirmCore index core$_\lambda(v)$ for each $v \in V$}
        \For{$v \in V$}
            \State $I[v] \leftarrow \text{Top-$\lambda$}(\text{deg}(v))$ \Comment{$\{$Fact~\ref{fact:1}$\}$}
            \State $B[I[v]] \leftarrow B[I[v]] \cup \{v\}$
        \EndFor
        \For{$k = 1, 2, \dots, |V|$}
            \While{$B[k] \neq \emptyset$}
                \State pick and remove $v$ from $B[k]$
                \State core$_\lambda(v) \leftarrow k$, $N \leftarrow \emptyset$
                \For{$(v, u, \ell) \in E$ and $I[u] > k$}
                    \State $\text{deg}_{\ell}(u) \leftarrow \text{deg}_{\ell}(u) - 1$
                    \If{$\text{deg}_{\ell}(u) = I[u] - 1$}\Comment{$\{$Fact~\ref{fact:2}$\}$}
                        \State $N \leftarrow N \cup \{u\}$
                    \EndIf
                \EndFor
                \For{$u \in N$}
                    \State remove $u$ from $B[I[u]]$
                    \State update $I[u]$ \Comment{$\{$By hybrid approach$\}$}
                    \State $B[I[u]] \leftarrow B[I[u]] \cup \{u\}$
                \EndFor
                \State $V \leftarrow V \:\char`\\ \: \{v\}$
            \EndWhile
        \EndFor
    \end{algorithmic}
\end{algorithm}


Algorithm \ref{alg:FirmCore} processes the nodes in increasing order of Top$-\lambda$ degree, by using a vector $B$ of lists such that each element $i$ contains all nodes with Top$-\lambda$ degree equal to $i$. Based on this technique, we keep nodes sorted throughout the algorithm and can update each element in $\mathcal{O}(1)$ time. After initializing the vector $B$, Algorithm \ref{alg:FirmCore} starts processing of $B$'s elements in increasing order, and if a node $v$ is processed at iteration $k$, its core$_\lambda$ is assigned to $k$ and removed from the graph. In order to remove a vertex from a graph, we need to update the degree of its neighbors in each layer, which leads to changing the Top$-\lambda$ degree of its neighbor, and changing their bucket accordingly. A straightforward approach is to update the Top$-\lambda$ degree and bucket of a node when we update its degree in a layer; however, it is inefficient since it re-computes the Top$-\lambda$ degree of some nodes without any changes in their Top$-\lambda$ degree. To address this issue, we use the following fact:
\vspace*{-1mm}
\begin{fact}
\label{fact:2}
Removing $v \in V$ cannot affect Top$-\lambda(deg(u))$ for $u \in N_\ell(v)$, unless before removing $v$, $\text{deg}_\ell (u) = $Top$-\lambda($deg$(u))$. 
\end{fact}
\vspace*{-1mm}
 Therefore, in lines 10 and 11 of the algorithm we store all potential nodes such that their Top$-\lambda$ degree can be changed, and in lines 12-15 we update their Top$-\lambda$ degree and their buckets. To improve the efficiency of the algorithm we use the fact that for each node in $N$, the Top$-\lambda$ degree can be unchanged or decreases by 1. The first method of updating the $I[u]$ is to check all deg$_\ell(u)$ to see whether there are $\lambda$ values bigger than $I[u]$, and if so, $I[u]$ does not need to be  updated and it is decremented  by one otherwise. The second method is to re-compute the Top$-\lambda$ degree each time. While the first method needs $\mathcal{O}(|L|)$ time, the second approach takes $\mathcal{O}(\lambda \log |L|)$. Therefore, since for $\lambda \geq \frac{|L|}{\log |L|}$ (resp. for $\lambda < \frac{|L|}{\log |L|}$)  the first method (resp. the second method)  is more efficient, we use a hybrid approach to improve the performance.

Algorithm \ref{alg:FirmCore} finds the FirmCore indices ($core_\lambda$) for a given $\lambda$, but for FirmCore decomposition we need to find all $core_\lambda$s for all possible values of $\lambda$. A possible way is to re-run Algorithm \ref{alg:FirmCore} for each value of $\lambda$, but it is inefficient since we can avoid re-computing the Top$-\lambda$ degree in line 2. Alternatively, we first sort the vector degree of each node in decreasing order and when threshold is $\lambda$, we assign $\lambda$-th element of sorted vector degree to $I[v]$. Moreover, since running Algorithm \ref{alg:FirmCore} for different values of $\lambda$ is independent, we can use multi-processing programming, which makes~the~code~much~faster. 


\noindent
\textbf{Computational complexity.}
The time complexity of Algorithm~\ref{alg:FirmCore} is $\mathcal{O}((|E| + |V|)|L| + |V|\lambda \log |L|)$, so based on the efficient computing of Top$-\lambda$ degree for all values of $\lambda$, the time complexity of FirmCore decomposition is $\mathcal{O}(|E||L|^2 + |V||L|\log |L|)$. Using a heap, we can find the $\lambda$-th largest element of degree vector in $\mathcal{O}(|L|+\lambda \log|L|)$ time so lines 1-3 take $\mathcal{O}(|V|( |L| + \lambda\log |L|))$ time. We remove each vertex exactly one time from the buckets, taking time $\mathcal{O}(|V|)$, and each edge is counted once, taking time $\mathcal{O}(|E|)$ (lines 8-11). Finally, to update the buckets for a given node we need $\mathcal{O}(|L|)$,~so~line~14~takes~$\mathcal{O}(|E| |L|)$~time.

\begin{algorithm}[t]
    \small
    \caption{$(k, r, \lambda)$-FirmD-Core}
    \label{alg:FirmD-core}
    \begin{algorithmic}[1]
        \Require{A directed ML graph $G = (V, E, L)$, and $\lambda \in \mathbb{N}^+$}
        \Ensure{FirmD-Core index $(k, r)$ for each pair of $k$ and $r$.}
        \ForAll{$k = 1, 2, \dots, |V|$}
            \State re-compute degrees
            \State $S \leftarrow \{v | v \in V, \text{and} \:\: \text{Top-}\lambda(\text{deg}^+(v)) \geq k \}, T \leftarrow V$
            \ForAll{$v \in S$}
                \State $I^+[v] \leftarrow \text{Top-$\lambda$}(\text{deg}^{H,+}(v))$
            \EndFor
            \ForAll{$u \in T$}
                \State $I^-[u] \leftarrow \text{Top-$\lambda$}(\text{deg}^{H,-}(u))$
                \State $B[I^-[u]] \leftarrow B[I^-[u]] \cup \{u\}$
            \EndFor
            \ForAll{$r = 1, 2, \dots, |V|$}
                \While{$B[r] \neq \emptyset$}
                    \State pick and remove $v$ from $B[r]$ and $T$
                    \State core$^k_\lambda(v) \leftarrow r$, $N \leftarrow \emptyset$
                    \ForAll{$(u, v, \ell) \in E$ s.t. $u \in S$ and $I^+[u] \geq k$}
                        \State $\text{deg}^{H,+}_\ell(u) \leftarrow \text{deg}^{H,+}_\ell(u) - 1$
                        \If{$\text{deg}^{H,+}_\ell(u) = I^+[u]$ - 1}
                            \State $N \leftarrow N \cup \{u\}$
                        \EndIf
                    \EndFor
                    \ForAll{$u \in N$}
                        \State update $I^+[u]$
                        \If{$I^+[u] < k$}
                            \State remove u from $S$
                            \ForAll{$(u, v', \ell) \in E$ and $I^-[v'] \geq r$}
                                \State $\text{deg}^{H,-}_\ell(v') \leftarrow \text{deg}^{H,-}_\ell(v') - 1$
                                \If{$\text{deg}^{H,-}_\ell(v') = I^-[v']$ - 1}
                                    \State remove $v'$ from $B[I^-[v']]$ 
                                    \State update $I^-[v']$ 
                                    \State $B[I^-[v']] \leftarrow B[I^-[v']] \cup \{v'\}$
                                \EndIf
                            \EndFor
                        \EndIf
                    \EndFor
                \EndWhile
            \EndFor
        \EndFor
    \end{algorithmic}
\end{algorithm}

\vspace{-2ex} 
\subsection{FirmD-Core Decomposition}
\label{sec:alg-FirmDCore}
A naive approach for FirmD-Core decomposition is to fix $k$ and $r$, and then initialize $(S, T)$-induced subgraph to $S = T = V$ and then recursively remove nodes whose outdegree in less than $\lambda$ layers are $\geq k$ from $S$, and nodes with indegree in less than $\lambda$ layers are $\geq r$ from $T$. Then we can do this  for all possible values of $k, r \in \mathbb{N}^+$. This algorithm is inefficient since each core is computed starting from the scratch. To address this, we use the containment property of FirmD-Core. If we fix the value of $k$, the FirmD-Cores for different values of $r$ are nested. Accordingly, we can repeat the above algorithm but assign the $(k, r)$-core index to each removed node during $r$-th iteration.

Algorithm \ref{alg:FirmD-core} shows the process of FirmD-Core decomposition. In this algorithm, $H$ refers to $(S, T)-$induced subgraph. In the first loop, we iterate over all possible values of $k$ and when $k$ is fixed, first we initialize $(S, T)$-induced subgraph with $S = T = V$. Lines 6-8 of the algorithm initialize the buckets such that in bucket $i$, all nodes with Top$-\lambda$ in-degree equal to $i$ are stored. Then similar to Algorithm~\ref{alg:FirmCore}, for each $r$, it removes nodes with in-degree less than $r$ from $T$ and then updates the degrees of its neighbors in $S$ and removes them from $S$ if their current out-degree is less than $k$.  

\noindent
\textbf{Computational complexity.}
While $k$ is fixed, the complexity of Algorithms \ref{alg:FirmCore} and \ref{alg:FirmD-core} are similar, so the time complexity of Algorithm \ref{alg:FirmD-core} is $\mathcal{O}(k_{\max}(|V|\lambda \log |L| + |V||L| + |E||L|))$, where $k_{\max}$ is the maximum Top-$\lambda$ out-degree of any node.

\vspace{-1ex}
\section{Multilayer Densest Subgraph}
\label{sec:Densest_Subgraph}
In this section, we show how to use the FirmCore and FirmD-Core concepts to devise an efficient approximation algorithm for the densest subgraph in ML graphs and directed ML graphs.

\vspace*{-1ex}
\subsection{Problem Definition}
\label{sec:MDS-definition}
We first recall the densest subgraph problem in undirected ML graphs proposed in \cite{MLcore}, and then extend it to  a formulation of the densest subgraph problem in directed ML networks. 

In contrast with single-layer graphs, the density objective function in ML graphs  provides a trade-off between the number of layers and the edge density. Accordingly, in the following definition, there is a parameter $\beta$ to control the importance of edge density in comparison to the number of layers exhibiting the~high~density.

\begin{problem} [Multilayer Densest Subgraph]\cite{MLcore, ml-core-journal}
Given an undirected ML graph $G=(V,E,L)$, a positive real number $\beta$, and a real-valued function $\rho : 2^V \rightarrow \mathbb{R}^+$ defined as:
\begin{equation} \label{eq1}
\rho(S) = \max_{\hat{L} \subseteq L} \min_{\ell \in \hat{L}} \frac{|E_\ell[S]|}{|S|} |\hat{L}|^\beta,
\end{equation}
find a subset of vertices $S^* \subseteq V$ that maximizes $\rho$ function, i.e., 
\begin{equation} \label{eq2}
S^* = \arg \max_{S \subseteq V} \rho(S).
\end{equation}
\end{problem}

Inspired by Problem 3 and the definition of the densest subgraph problem in directed single-layer graphs \cite{DirectedDensity}, we define the directed multilayer densest subgraph problem as follows.

\begin{problem} [Directed Multilayer Densest Subgraph]
Given a directed ML graph $G=(V,E,L)$, a positive real number $\beta$, and sets of vertices $S,T \subseteq V$, which are not necessarily disjoint, let $E[S, T]$ be the set of edges linking vertices in $S$ to the vertices in $T$. Given a real-valued function $\rho : 2^V \times 2^V \rightarrow \mathbb{R}^+$ defined as:
\begin{equation} \label{eq3}
\rho(S, T) = \max_{\hat{L} \subseteq L} \min_{\ell \in \hat{L}} \frac{|E_\ell(S, T)|}{\sqrt{|S| |T|}} |\hat{L}|^\beta,
\end{equation}
find subsets of vertices $S^*, T^* \subseteq V$ that maximize $\rho(.)$ function, i.e., 
\begin{equation} \label{eq4}
(S^*, T^*) = \arg \max_{S, T \subseteq V} \rho(S, T).
\end{equation}
\end{problem}

\subsection{Algorithms}
\label{sec:MDS-algorithm}
\subsubsection{Undirected Multilayer Densest Subgraph} 
Finding the densest subgraph in ML networks is NP-hard \cite{ml-core-journal}, and to the best of our knowledge, there is no known polynomial-time approximation algorithm for this problem. The state-of-the-art approximation algorithm for the ML densest subgraph problem is based on the $\textbf{k}$-core decomposition \cite{MLcore}. It computes the ML core decomposition of the input graph $G$ and, among the cores found, it returns the one that maximizes the objective function $\rho(.)$. Although it achieves a $\frac{1}{2|L|^{\beta}}$-approximation, it is exponential-time, and thus cannot scale to  large graphs. More precisely, the computational complexity of the ML core decomposition algorithm in \cite{ml-core-journal} is $\mathcal{O}(|V|^{|L|} (|E| + |V||L|))$ in the worst case, which makes their approximation algorithm impractical for large networks even for a small number of layers. Here we propose the first polynomial-time approximation algorithm with a time complexity of $\mathcal{O}(|E| |L|^2 + |V||L|\log |L|)$ that has an approximation guarantee of $\frac{\psi_\beta}{2|L|^{\beta + 1}}$, where for a given $\beta$, $\psi_\beta$ is a function of the network topology, which ranges from $1$ to a number very close to $|L|^{\beta + 1}$. Therefore, not only is our algorithm much faster but also in many cases achieves a better quality.

\begin{algorithm}[t]
    \small
    \caption{FC-Approx}
    \label{alg:Approx1}
    \begin{algorithmic}[1]
        \Require{A multilayer graph $G = (V, E, L)$ and $\beta \in \mathbb{R}^+$}
        \Ensure{The densest FirmCore $C^*$ of $G$}
        \State $\Omega \leftarrow$ FirmCoreDecomposition(G)
        \State \Return $\arg \max_{C_{k, \lambda} \in \Omega} \rho(C_{k, \lambda})$
    \end{algorithmic}
\end{algorithm}

Algorithm \ref{alg:Approx1} is devised based on the FirmCore decomposition of a graph. The procedure starts with computing all $(k, \lambda)$-FirmCores for all possible values of $k$ and $\lambda$, and returns the FirmCore that maximizes the $\rho(.)$ function. The computational complexity is essentially dominated by that of FirmCore Decomposition.

\noindent
\textbf{Approximation Guarantees.} We first establish some lemmas based on the definition of FirmCore and properties of objective function $\rho(.)$, which we subsequently use to prove the approximation guarantee for Algorithm \ref{alg:Approx1}. Let  $S^*_{\text{SL}}$ denote the densest subgraph among all single-layer densest subgraphs, i.e., $S^*_{\text{SL}} = \arg \max_{S\subseteq V} \max_{\ell \in L} \frac{|E_\ell[S]|}{|S|}$, and let $\ell^*$  denote the layer where $S^*_{\text{SL}}$  exhibits its largest single-layer density. For $S\subseteq V$, let $\mu(S, \ell)$ denote the minimum degree of the (induced) subgraph $S$ in layer $\ell$. Furthermore, let $C^+$ be the non-empty $(k^+, \lambda^+)$-FirmCore of $G$ with maximum value of $\lambda^+$ such that $k^+ \geq \mu(S^*_{\text{SL}}, \ell^*)$. It is easy to see that $C^+$ is well defined. This FirmCore will help us to provide a connection between the upper bound on the density of the optimal solution, which is based on $S^*_{\text{SL}}$, and the lower bound on the density of the output of Algorithm~\ref{alg:Approx1}. Finally, let $C^*$ and $S^*$  denote the output of Algorithm \ref{alg:Approx1} and the optimal solution, respectively.

\setcounter{theorem}{0}
\begin{lemma}\label{lemma1}
Let $C$ be the $(k, \lambda)$-FirmCore of $G$, we have:
\begin{equation*}
\rho(C) \geq \frac{k}{2 |L|} \underset{\xi \in \mathbb{Z}, 0\leq \xi < \lambda}{\max} (\lambda - \xi) (\xi + 1)^{\beta}.
\end{equation*}

\end{lemma}
\begin{proof}
By definition, each node $v \in C$ has at least $k$ neighbors in at least $\lambda$ layers, so based on the pigeonhole principle, there exists a layer such that there are $\geq \frac{\lambda |C|}{|L|}$ nodes that have $\geq k$ neighbors. Let $\ell'$ denote  this layer, so we have:
$$\frac{|E_{\ell'} [C]|}{|C|} \geq \frac{\frac{\lambda |C|}{|L|} \times k}{2|C|} = \frac{\lambda k}{2|L|}.$$
Now,  ignoring this layer, and exploiting  the definition of $C$, and re-using the pigeonhole principle, we can conclude that there exists a layer such that there are $\geq \frac{(\lambda - 1) |C|}{|L|}$ nodes that have $\geq k$ neighbors. So there is a subset of $L$ with two layers such that each layer has density at least $\frac{(\lambda - 1) k}{2|L|}$. By iterating this process, we can conclude that $\exists \Tilde{L} \subseteq L$, such that: $$\rho(C) \geq \frac{(\lambda - |\Tilde{L}| + 1)k}{2|L|} |\Tilde{L}|^\beta.$$
The right hand side of the above inequality is a function of $|\Tilde{L}|$, and since the above inequality is valid for each arbitrary integer $1 \leq |\Tilde{L}| \leq \lambda$, $\rho(C)$ is not less than the maximum of this function exhibited by these values.
\end{proof}

\begin{lemma}\label{lemma2}
$\rho(C^*) \geq \frac{\psi_\beta}{2 |L|}  \mu(S^*_{\text{SL}}, \ell^*)$, where $\psi_{\beta} =\underset{\underset{0 \leq \xi <\lambda^+}{\xi \in \mathbb{Z}}}{max} (\lambda^+ - \xi)~(\xi~+~1)^\beta$.
\end{lemma}

\vspace{-5ex}
\begin{lemma}\label{lemma3}
$\rho(S^*) \leq \mu(S^*_{\text{SL}}, \ell^*) |L|^\beta.$
\end{lemma}

We now establish the approximation guarantee of Algorithm~\ref{alg:Approx1}. 
 \setcounter{theorem}{0}
\begin{theorem}\label{theorem1}
Let $\psi_{\beta} =\underset{\xi \in \mathbb{Z}, 0 \leq \xi  <\lambda^+}{max} \:   (\lambda^+ - \xi) (\xi + 1)^\beta$, we have: 
\[
\rho(C^*) \geq \frac{\psi_\beta}{2 |L|^{\beta + 1}} \rho(S^*).
\]

\end{theorem}

The function $f(\xi) = (\lambda^+ - \xi)(\xi+1)^\beta$ achieves its maximum value at $\xi = \frac{\lambda^+ \beta - 1}{\beta + 1}$, but since we want to maximize $f(.)$ over integers, 
the maximum occurs at $\xi = \floor*{\frac{\lambda^+ \beta - 1}{\beta + 1}}$ or $\xi = \ceil*{\frac{\lambda^+ \beta - 1}{\beta + 1}}$. Notice that a lower bound of $\psi_\beta$ can be obtained by setting  $\xi = \lambda^+ - 1$ or $\xi = 0$, yielding  $\psi_\beta \geq \max\{(\lambda^{+})^\beta, \lambda^+\}$. The following example  illustrates the  approximation factor of FC-Approx (Algorithm~\ref{alg:Approx1}) in relation to  the state-of-the-art, Core-Approx \cite{MLcore}.

\begin{example}
In Figure \ref{fig:example}, $|L| = 3$, blue nodes are $S^*_{\text{SL}}$ and $\ell^* = 1$. It is easy to see that $\mu(S^*_{\text{SL}}, \ell^*) = 4$ and $\lambda^+ = 2$. For $\beta = 1, 2, 3$, FC-Approx (resp. Core-Approx) achieves an approximation factor of $\frac{1}{9}, \frac{2}{27}, \frac{4}{81}$ (resp.  $\frac{1}{6}, \frac{1}{18}, \frac{1}{54}$). For $\beta \leq 1$, FC-Approx has a slightly worse approximation factor, while for $\beta > 1$, it guarantees better quality. 
How much worse (resp. better) can FC-Approx get compared to Core-Approx? To see that, consider 
removing any edge from the blue subgraph in third layer (``Neural Network''). This leads to $\lambda^+ = 1$ and $\psi_{\beta}=1$, which is the worst case of Algorithm \ref{alg:Approx1}, leading to an approximation factor of $\frac{1}{2\cdot3^{\beta+1}}$ as against $\frac{1}{2\cdot3^{\beta}}$  for Core-Approx. On the other hand, consider instead removing any edge from the blue subgraph in the first layer (``Algorithm''). This leads to $\lambda^+ = |L| = 3$, which is the best case of Algorithm \ref{alg:Approx1} and in this case, when $\beta \geq 1$ (resp. $\beta \leq 1$), $\psi_{\beta} \geq 3^{\beta}$ (resp. $\psi_\beta \geq 3$), leading to an approximation factor of  $\frac{1}{2\cdot3}$  (resp. $\frac{1}{2\cdot3^{\beta}}$) as against $\frac{1}{2\cdot3^{\beta}}$  for Core-Approx. Thus, the guaranteed approximation quality of FC-Approx can be up to $3\times$ worse than Core-Approx in the worst case. In the best case, when $\beta > 1$ (resp. $\beta \leq 1$), it can be $3^{\beta-1}\times$ better than  (resp. matches) that of Core-Approx. 
\end{example}

\subsubsection{Directed Multilayer Densest Subgraph}
\label{sec:FDC-Approx}
Similar to the undirected graphs, FDC-Approx Algorithm computes all $(k, r, \lambda)$-FirmD-cores of the graph and returns the FirmD-core that maximizes the objective function (pseudocode in Appendix~\ref{sec:FDC-Approx-appendix}).

\noindent
\textbf{Approximation Guarantees.}  Let $G = (V, E, L)$ be a directed ML graph, $(S^*, T^*)$ be the optimal solution, $(S^*_{\text{SL}}, T^*_{\text{SL}})$ be the densest subgraph among all single-layer densest subgraphs, and $\ell^*$ be its corresponding layer. The notion of $[x, y]$-core in single-layer networks~\cite{xy-core} corresponds to the Definition~\ref{FirmDCore}, when $L = 1$ and $\lambda = 1$ and $[x^*, y^*]$ is called the maximum cn-pair if $x^* \cdot y^*$ achieves the maximum value among all the possible non-empty $[x, y]$-cores. We use $[x^*, y^*]$ to denote the maximum cn-pair in $\ell^*$. Finally, $G[\hat{S}, \hat{T}]$ refers to the non-epmty $(\hat{k}, \hat{r}, \hat{\lambda})$-FirmD-Core with maximum $\hat{\lambda}$ that $\hat{k}\hat{r} \geq x^*y^*$. It is easy to see that this FirmD-Core is well defined.

\setcounter{theorem}{3}
\begin{lemma}
\label{lem:lemma5}
Given an arbitrary $(k, r, \lambda)$-FirmD-core, $H = G[S, T]$:
\begin{equation*}
    \rho(S, T) \geq \underset{\xi \in \mathbb{Z}, 0\leq \xi < \lambda}{\max} \frac{(\lambda - \xi) (\xi + 1)^{\beta}}{|L|} \times \max \{k \sqrt{a}, \frac{r}{\sqrt{a}}\},
\end{equation*}
where $a = \frac{|S|}{|T|}$.
\end{lemma}

 \setcounter{theorem}{1}
\begin{theorem}\label{theorem2}
Let $(S^+, T^+)$ represent the output of FDC-Approx, then
$\rho(S^+, T^+) \geq  \frac{\psi_\beta}{2|L|^{\beta + 1}} \rho(S^*, T^*)$, where $\psi_{\beta} =\underset{\xi \in \mathbb{Z}, 0 \leq \xi <\hat{\lambda}}{max} \: (\hat{\lambda} - \xi) (\xi + 1)^\beta$.
\end{theorem}

\vspace{-3ex}
\subsection{Discussion}
Another well-studied densest subgraph problem in single-layer graphs is to find the maximum $k$-core, which is equivalent to finding the subgraph $S^*$ that maximizes the minimum induced degree. Based on this density measure, several community models are introduced~\cite{community1, community2}. The extended version of this problem to ML networks, \textit{Best Friends Forever (BFF)} problem \cite{BFF}, can lead to more powerful community models than their single-layer counterparts. 
\vspace{-1ex}
\begin{problem}[BFF-MM]\cite{BFF} 
\label{prob:bff}
Given an ML network $G = (V, E, L)$,  let $\mu(S, \ell)$ denote the minimum induced degree in $S$ in layer $\ell$. Find a subset of vertices $S^* \subseteq V$ such that:
\begin{equation}
\label{eq6}
    S^* = \arg \max_{S \subseteq V} \min_{\ell \in L} \mu(S, \ell).
\end{equation}
\end{problem}

\begin{theorem}
Given an ML network $G$, let $k_{\max}$ be the maximum value of $k$ such that the $(k, |L|)$-FC of $G$ is non-empty. $(k_{\max}, |L|)$-FC is the exact solution to Problem~\ref{prob:bff}.
\end{theorem}
\vspace{-2ex}

\begin{figure*}
    \centering
    \hspace{-1ex}
    \includegraphics[height=0.13\textwidth, width=\textwidth]{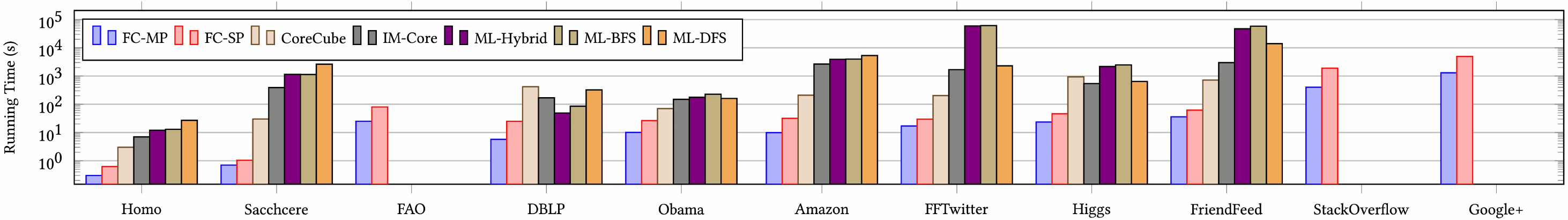}
     \vspace{-7mm}
    \caption{Comparative evaluation: the runtime of the proposed algorithms, ML-core algorithms, and CoreCube algorithm.}
    \label{fig:Compare}
\end{figure*}

\begin{center}
\begin{table} [tpb!]
 \caption{Network Statistics}
 \vspace{-2ex}
\resizebox{0.18\textwidth}{!} {
\subfloat[][Undirected Networks]{
\begin{tabular}{l | c c c}
 \toprule
  {Dataset} & {$|V|$} &  {$|E|$} & {$|L|$}\\
 \hline\hline 
    Homo&  18k  &  153k  &  7   \\
    Sacchcere & 6.5k  &  247k  &   7   \\
    DBLP &  513k &  1.0M  &  10   \\
    Amazon & 410k  &  8.1M  &  4   \\
    FFTwitter & 155k  &  13M  &   2   \\
    Friendfeed & 510k  &  18M  &  3    \\
    StackOverflow &  2.6M  &  47.9M  &   24   \\
    Google+ &  28.9M  &  1.19B  &  4    \\
 \bottomrule
\end{tabular}
}
}
\hspace{4mm}
\resizebox{0.161\textwidth}{!} {
\subfloat[][Directed Networks]{
\begin{tabular}{l | c c c}
 \toprule
  {Dataset} & {$|V|$} &  {$|E|$} & {$|L|$}\\
 \hline\hline
    Slashdot  &  51k &  139k  &   11 \\
    FAO & 214  &  319k  &   364   \\
    Sanremo   & 56k  &  462k  &   3  \\
    Cannes   & 438k  &  1.2M  &  3   \\
    Diggfriends & 279k  &  1.7M  &  4   \\
    UCL & 677k  &  1.7M  &  3   \\
    Obama & 2.2M  &  3.8M  &   3   \\
    Higgs & 456k  &  13M  &  4   \\
 \bottomrule
\end{tabular}
}
}
 \label{tab:datastat}
\end{table}
\end{center}

\vspace*{-2ex}
\section{Experiments}
\label{sec:Experiments}

In this section, we \textbf{(1)} evaluate the efficiency of our algorithms and compare them with the state-of-the-art ML core decompositions,
\textbf{(2)} compare the cohesiveness of FirmCore and \textbf{k}-core,
\textbf{(3)} explore the characteristics of FirmCore, \textbf{(4)} evaluate and compare the quality of the solutions found by FC-Approx and state-of-the-art for the ML densest subgraph problem, \textbf{(5)} evaluate the solution of FDC-Approx to the directed ML densest subgraph problem, and \textbf{(6)} show the effectiveness of FirmCore and FirmD-Core via case studies.

\noindent
\textbf{Setup.} 
We perform extensive experiments on sixteen real networks from \cite{KONECT, MLcore, Friendfeed, Higgs, amazon_datset, Google+, FAO, homo, Twitter_datasets} including social, genetic, co-authorship, financial, and co-purchasing networks, whose main characteristics are summarized in Table~\ref{tab:datastat}. In some experiments, we have ignored the direction of edges and treated directed graphs as undirected. All methods are implemented in Python (Appendix~\ref{sec:Reproducibility}). The experiments are performed on a Linux machine with Intel Xeon 2.6 GHz CPU and 128 GB RAM, using 4 cores in case of  multi-core processing implementations, and using a single core otherwise.

\noindent
\textbf{Efficiency.}
We first compare the FirmCore (FC for short) decomposition algorithm  with other dense structure mining algorithms in ML networks, which are ML-Hybrid, ML-BFS, ML-DFS \cite{MLcore, ml-core-journal}, IM-Core \cite{ml-core-journal}, and CoreCube \cite{CoreCube}, in terms of the running time and memory usage. Since our FC algorithm is capable of running in a multi-processor setting, we use both settings, multi-processor and single-processor (denoted FC-MP and FC-SP). Notably, other algorithms are not capable of running in multiprocessor setting. The running time results are shown in Figure~\ref{fig:Compare}. FirmCore decomposition, even FC-SP, outperforms the other five algorithms in all datasets. In large networks, FC  is at least $70\times$  (resp. $10\times$) faster than the ML-Hybrid, ML-DFS, and ML-BFS algorithms (resp. CoreCube and IM-Core), with the significant improvement being for the massive datasets. While FC-SP (resp. FC-MP) takes less than 1.5 hours (resp. half hour) on Google+ dataset, other ML $\textbf{k}-$core decomposition algorithms do not terminate in 10 days. Similar results are observed on small datasets like FAO with a large number of layers:  here, FC-SP (resp. FC-MP) terminates in under 2 minutes (resp. half a minute). Another important characteristic  of core decomposition methods is their memory usage: FirmCore outperforms other algorithms in terms of memory usage (Appendix~\ref{sec:Memory}). Notice that missing bars in Figure~\ref{fig:Compare} correspond to either taking more than 10 days or memory usage~more~than~100GB. 

\begin{table} [tpb!]
\vspace{-2ex}
 \caption{FirmD-Core evaluation}
 \vspace{-2ex}
\resizebox{0.51\textwidth}{!} {
\hspace{-9mm}
\begin{tabular}{c | c  c | c c | c c | c c | c c | c c | c c | c c}
 \toprule
  & \multicolumn{2}{c|}{Slashdot}   & \multicolumn{2}{c|}{FAO}  & \multicolumn{2}{c|}{Sanremo} & \multicolumn{2}{c|}{Cannes} & \multicolumn{2}{c|}{Diggfriends} &  \multicolumn{2}{c|}{UCL}  &  \multicolumn{2}{c|}{Obama} & \multicolumn{2}{c}{Higgs} \\
 \hline \hline 
    \# Cores  &  \multicolumn{2}{c|}{1037}  &  \multicolumn{2}{c|}{106474}    &     \multicolumn{2}{c|}{4614}      &   \multicolumn{2}{c|}{4407} & \multicolumn{2}{c|}{55061} & \multicolumn{2}{c|}{3387} & \multicolumn{2}{c|}{54888} & \multicolumn{2}{c}{258402}   \\
    \midrule
    Method  &  SP & MP  &  SP & MP    &     SP & MP      &   SP & MP & SP & MP & SP & MP & SP & MP & SP & MP   \\
\midrule
Time (s) & 117 & 39 & 2035 & 501 & 105& 75& 2430& 674& 3089 & 2024 & 499 & 351 & 194681 & 132112 & 203578 & 143749\\
\midrule  
Memory (MB)& 150 & 151 & 36 & 37 & 89 & 91 & 594 & 596 & 421& 424& 908 & 919 & 2892 & 2909 & 778&785\\
 \bottomrule
\end{tabular}
}
\vspace{-3ex}
\label{tab:FirmD-Dore}
\end{table}

Since there are no prior algorithms for core decomposition of  directed ML networks, we just compare the two settings of our algorithm, single-processor and multi-processors, shown in Table~\ref{tab:FirmD-Dore}. As expected, using multi-processors is much faster than a single-processor, and it takes  at most $1\%$ more memory. The speedup varies depending on \#layers and \#FirmCores found.

\begin{table} [tpb!]
 \caption{Average degree of each layer in densest FirmCore and densest \textbf{k}-core.}
 \vspace{-2ex}
\resizebox{0.48\textwidth}{0.09\textwidth}{
\begin{tabular}{l | c | c c c c c c c c}
 \toprule
  {Dataset ($\beta$)} & Model & {layer 1} &  {layer 2} & {layer 3} & {layer 4} &  {layer 5} & {layer 6} & {layer 7} & {Average}\\
 \hline\hline 
    \multirow{2}{*}{Higgs ($\beta = 3$)} & FirmCore & 19.2 & 9.88 & 8.94 & 0.56 & - & - & - & \textbf{9.64}\\
    & \textbf{k}-Core & 18.13 & 10.16 & 9.44 & 0.75 & - & - & - & 9.62\\
    \hline
    \multirow{2}{*}{Higgs ($\beta = 6$)} & FirmCore & 19.2 & 9.88 & 8.94 & 0.56 & - & - & - & \textbf{9.64}\\
    & \textbf{k}-Core & 10.38 & 5.27 & 5.62 & 2.3 & - & - & - & 5.89\\
    \hline
    \multirow{2}{*}{Sacchcere ($\beta = 1$)} & FirmCore & 15.2 & 10.68 & 2.71 & 15.2 & 0.36 & 0.28 & 33.27  & \textbf{11.1}\\
    & \textbf{k}-Core & 4.1 & 8.35 & 2.37 & 26.3 & 0.37 & 0.24 & 26.66 & 9.77\\
    \hline
    \multirow{2}{*}{Sacchcere ($\beta = 2$)} & FirmCore & 11.63 & 11.07 & 5.43 & 12.55 & 0.55 & 0.37 & 27.77  & \textbf{9.91}\\
    & \textbf{k}-Core & 10.17 & 7.49 & 7.45 & 8.77 & 0.54 & 0.4 & 15.11 & 7.13\\
        \hline
    \multirow{2}{*}{Homo ($\beta = 1$)} & FirmCore & 2.62 & 27.03 & 0.02 & 0.18 & 8.57 & 0.05 & 0.02 & \textbf{5.49}\\
    & \textbf{k}-Core & 2.17 & 28.51 & 0.02 & 0.21 & 5.89 & 0.06 & 0.02 & 5.26\\
    \hline
    \multirow{2}{*}{Homo ($\beta = 2$)} & FirmCore & 6.27 & 23.23 & 0.08 & 0.52 & 5.79 & 0.24 & 0.02 & \textbf{5.16}\\
    & \textbf{k}-Core & 6.24 & 15.48 & 0.09 & 0.72 & 6.36 & 0.2 & 0.02 & 4.15\\
 \bottomrule
\end{tabular}
}
 \label{tab:average_density}
\end{table}

\noindent
\textbf{Cohesiveness}. Next, we compare the cohesiveness of FirmCore and \textbf{k}-core \cite{azimi-etal}. As discussed in \S~\ref{sec:RW}, CoreCube \cite{CoreCube} is a special case of \textbf{k}-core, so by considering \textbf{k}-core decomposition we do not miss the most cohesive subgraphs found by CoreCube. We use two metrics. The first metric is the average degree density in each layer, and the second metric is local clustering coefficient (LCC). Since the numbers of cores are different, to have a fair comparison, here, we consider the densest FirmCore and the densest \textbf{k}-core subgraphs obtained by density measure $\rho$ in ML networks~\cite{MLcore}. We choose the value of $\beta$ such that the solution is stable (i.e., does not change for $\beta \pm \epsilon$). Table~\ref{tab:average_density} shows the results on Higgs, Sacchcere, and Homo datasets. Not only the average of layers' density is higher for densest FirmCore, but also it has a higher density in more layers.

Moreover, Figure~\ref{fig:CC_higgs_main} shows the distribution of nodes' LCC in densest FirmCore and densest \textbf{k}-core on the Higgs dataset. While the overall distributions are very similar, FirmCore is capable of finding a subgraph that contains more nodes with high LCC ($\geq 0.75$) and less nodes with low LCC ($\leq 0.1$). The results on other datasets are similar, which we omit  for lack of space.

\begin{figure}[tpb!]
    \centering
    \includegraphics[width=0.24\linewidth, height=0.185\linewidth
    ]{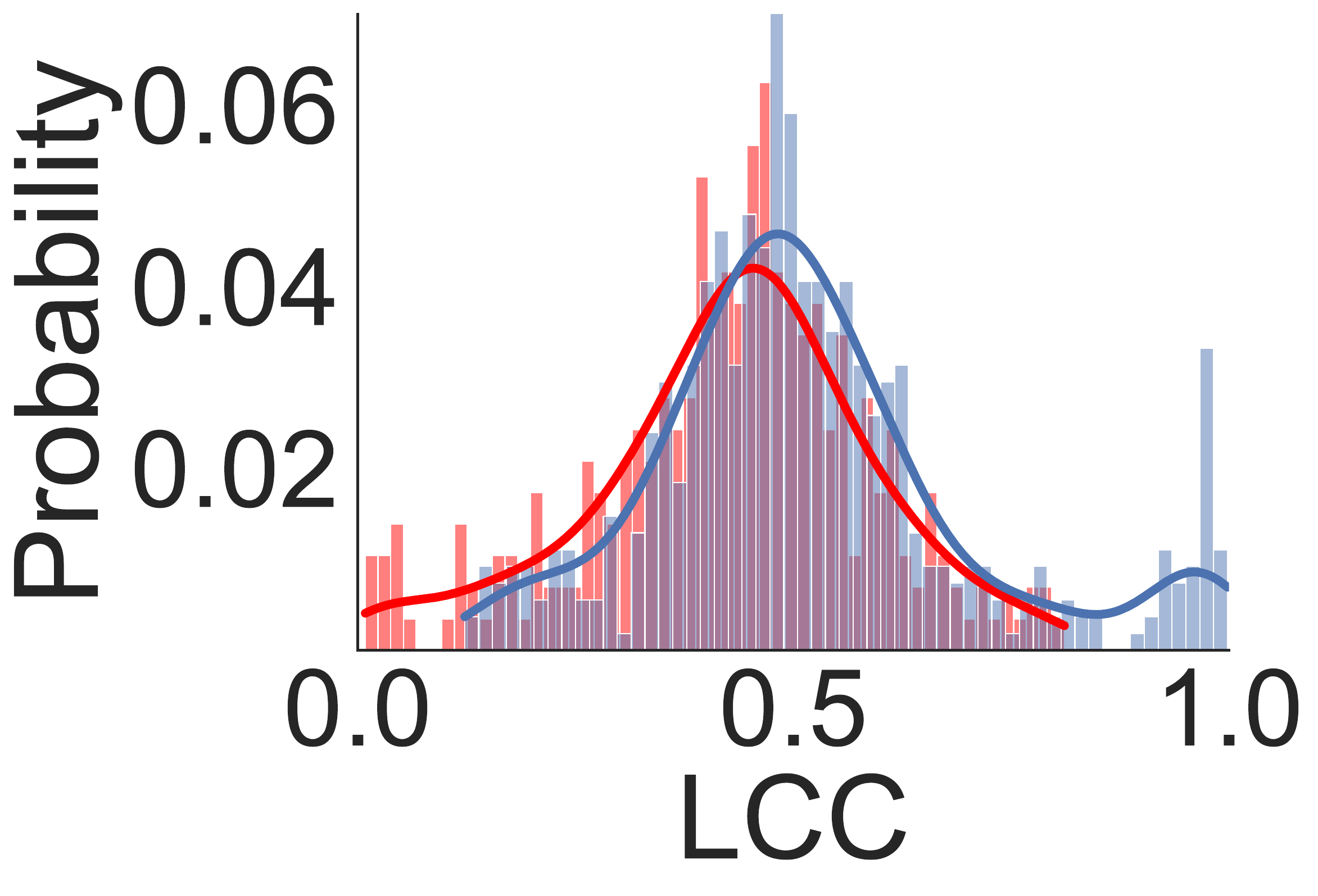}
    \includegraphics[width=0.24\linewidth, height=0.185\linewidth
    ]{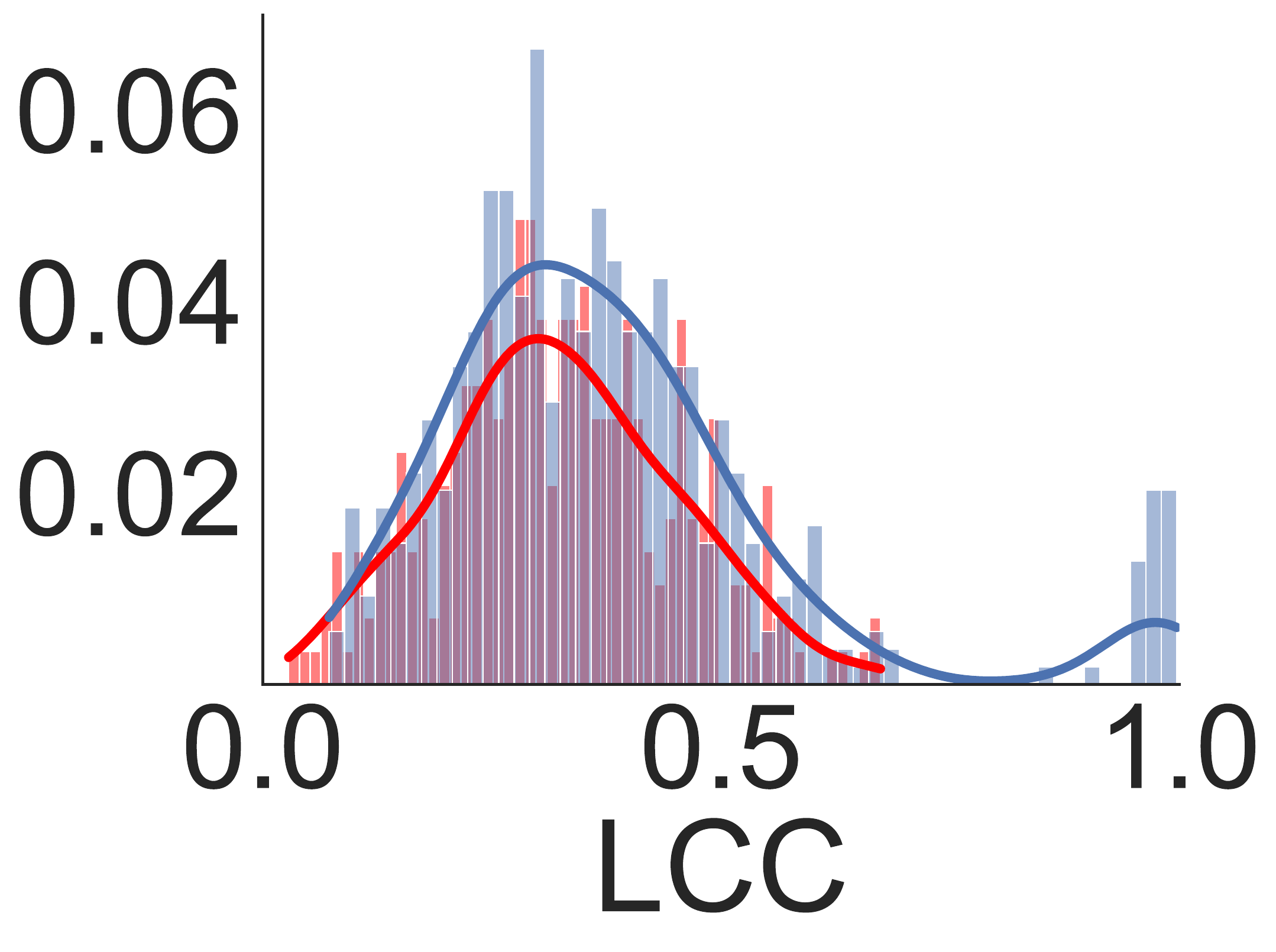}
    \includegraphics[width=0.24\linewidth, height=0.185\linewidth
    ]{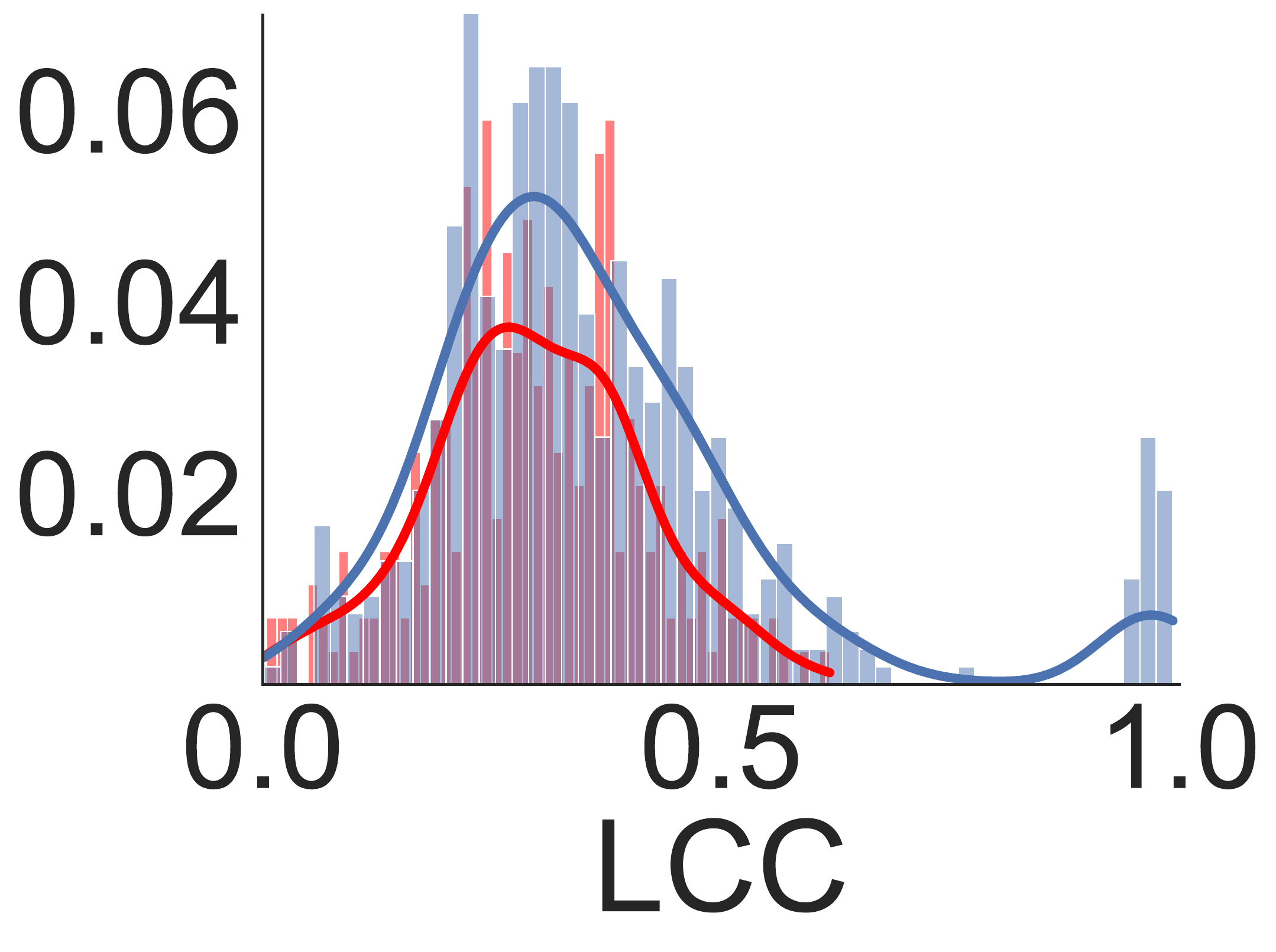}
    \includegraphics[width=0.24\linewidth, height=0.185\linewidth
    ]{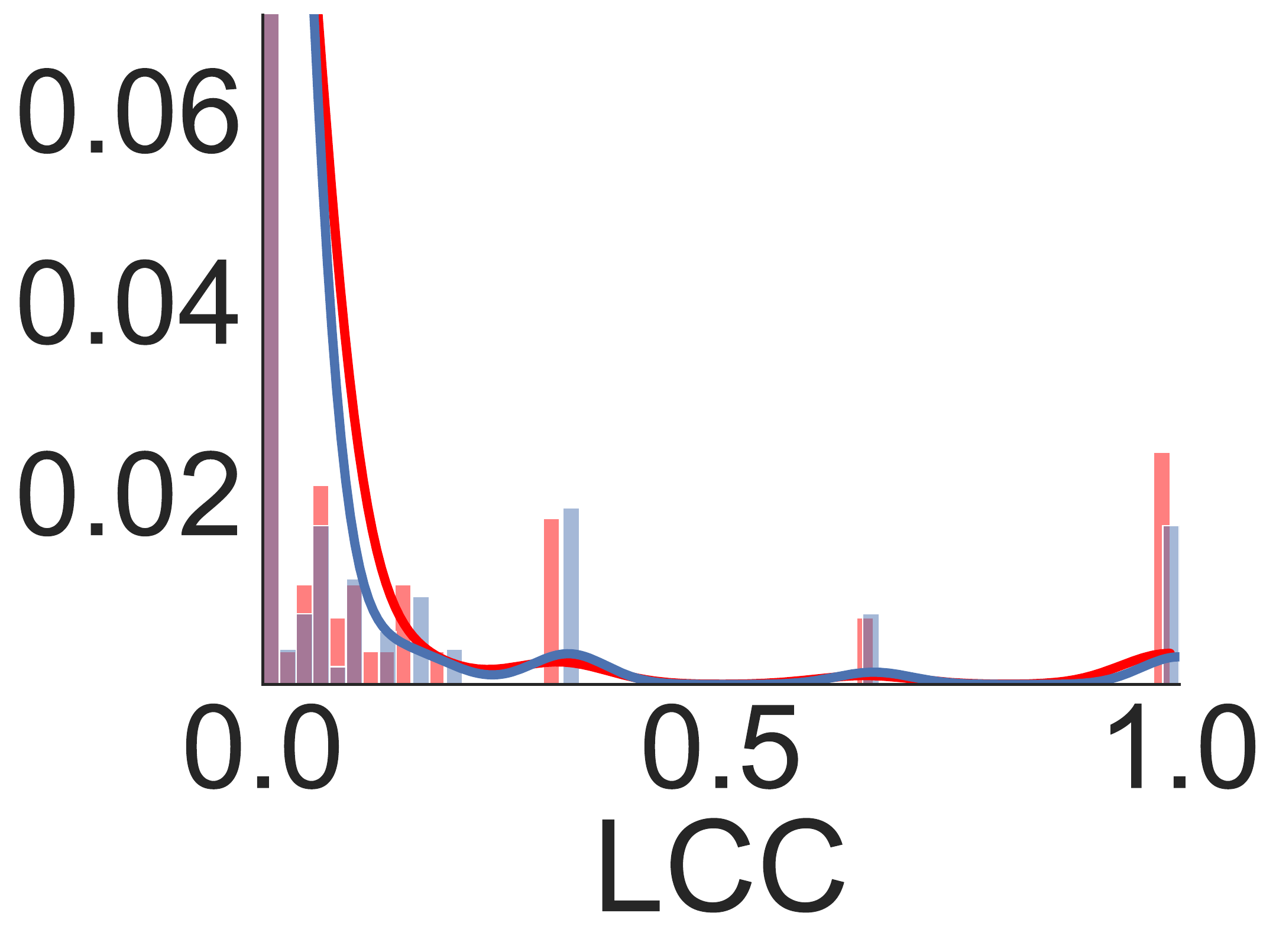}
    \vspace{-3ex}
    \caption{Distribution of LCC of densest FirmCore (Blue) and densest ML \textbf{k}-core (Red) on Higgs dataset ($\beta = 3$). Left to right, layer 1, 2, 3, and 4, respectively. Solid lines show the kernel density estimation.}
    \label{fig:CC_higgs_main}
    \vspace{-3ex}
\end{figure}

\begin{figure}[t!]
    \centering
    \includegraphics[width=0.375\textwidth]{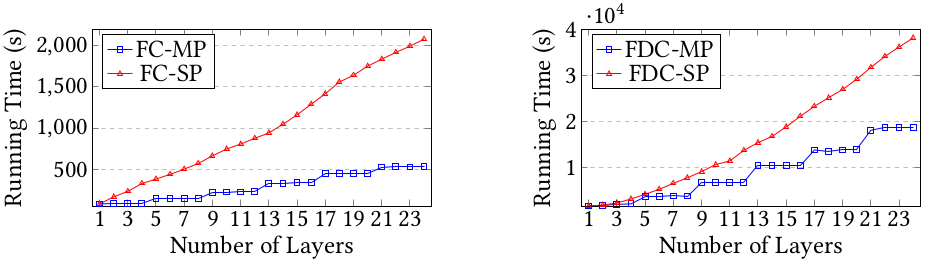}
    \vspace{-4mm}
    \caption{The runtime of the FirmCore and FirmD-Core algorithms with varying the number of layers (Left: StackOverflow, Right: DiggFriends).}
    \label{fig:Layers}
\end{figure}

\noindent
\textbf{Scalability.} Figure \ref{fig:Layers} demonstrates the effect of \#layers on the running time of the FirmCore and FirmD-Core algorithms. For the FirmCore algorithm, we use different versions of StackOverflow obtained by selecting a variable number of layers from 1 to 24, and use a similar approach on the Diggfriends network, used in evaluating FirmD-Core. Unlike previous  algorithms whose running time grows exponentially, the running time of  FirmCore, as well as FirmD-Core, scales gracefully. Thus, FirmCore and FirmD-Core can be applied on large networks with a large \#layers. Moreover, based on the running time results in Figure~\ref{fig:Compare} (resp. Table~\ref{tab:FirmD-Dore}), we see that FirmCore (resp. FirmD-Core) algorithm can scale to graphs containing billions (resp. tens of millions) of edges.

\begin{figure}
    \centering
    \includegraphics[width=0.5\textwidth]{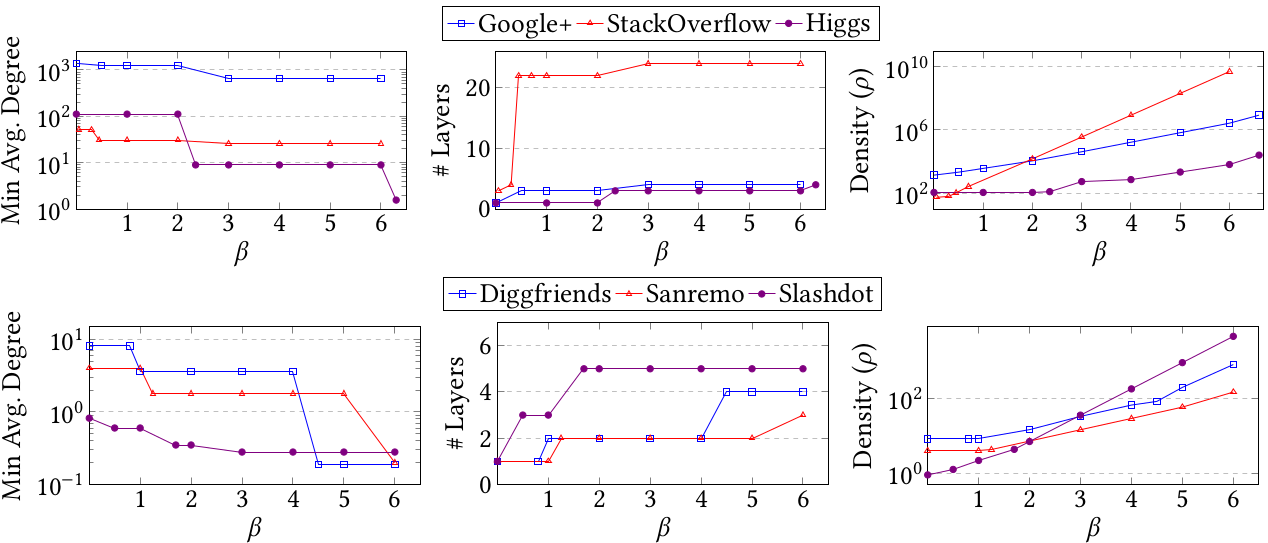}
    \vspace{-6ex}
    \caption{Minimum average-degree in a layer, number of selected layers, and $\rho$'s value of the output of the FC-Approx ($1^{st}$~row), and FDC-Approx ($2^{nd}$~row),~with~varying~$\beta$.}
    \label{fig:Densest-subgraph-stat}
\end{figure}

\noindent
\textbf{Density.}
While we guarantee an average degree density ($\rho$) for each FirmCore in Lemmas~\ref{lemma1} and \ref{lem:lemma5}, we experimentally evaluate the FirmCores' density using  clique-density measure. Results suggest while FirmCore captures several near-cliques (subgraphs  of density $\approx 1$), it finds large sets of lower density ($0.15-0.2$) as well, which is a significant density for sets of hundreds of nodes (Appendix~\ref{sec:density_analysis}).

\noindent
\textbf{Frequent quasi-cliques.}
As discussed before, the FirmCore decomposition can speed up the extraction of frequent quasi-cliques. The results suggest a significant pruning of the search space by FirmCore, which in the worst case (resp. best case) prunes 87$\%$ (resp. $99.99\%$) of the search space (Appendix~\ref{sec:frequent_quasi_clique}).

\noindent
\textbf{FC index and Top-$\lambda$ degree.}
As an important characteristic of the FirmCore, we found that in real networks, there is a strong correlation between the $core_\lambda$ index of a node and its Top-$\lambda$ degree, which is its upper bound. Similar results have been seen in single-layer networks, which have led to an anomaly detection approach \cite{anomaly-detection1}. The Spearman’s rank correlation coefficient between Top-$\lambda$ degree and FirmCore index on all datasets is significantly more than $0.7$ ($p$-value $< 0.0001$) and in most cases is close to 1 (Appendix~\ref{sec:corr}).

\vspace*{-2ex} 
\subsection{Multilayer Densest Subgraph}
In this part, we evaluate the quality of the solution proposed by FC-Approx and FDC-Approx for the densest subgraph problem. Figure~\ref{fig:Densest-subgraph-stat} provides characteristics of approximation solution of the ML densest subgraph with varying $\beta$. In a directed $(S, T)-$induced subgraph, we consider average degree as $\frac{|E(S, T)|}{\sqrt{|S|\cdot |T|}}$~\cite{DirectedDensity}. The results are expected as decreasing $\beta$ obliges the model to choose a subgraph with large average-degree density in a few layers and even one layer when $\beta$ is close enough to zero. That is why the minimum average degree has a decreasing trend with increasing $\beta$ and opposite the number of layers has increasing trends. On the other hand, density, $\rho(.)$, as a function of $\beta$ shows an exponential trend, as expected.

\noindent
\textbf{Comparative Evaluation.} 
Here we compare the quality of FC-Approx and state-of-the-art, Core-Approx~\cite{MLcore}, by evaluating~their solutions to the ML densest subgraph problem. Figure~\ref{fig:Compare_Densest} shows~the~result~of this comparison on DBLP, Amazon, and FriendFeed networks. On the DBLP dataset, FC-Approx outperforms Core-Approx and provides a solution with greater density $(\rho)$ for all values of $\beta$. In other two cases, the density values are close. Our results report that in the best case (resp. worst case), FC-Approx provides a solution with $2.44\times $ (resp. $\frac{1}{1.4} \times $) density of the Core-Approx's solution. As a conclusion, while the FirmCore-Approx algorithm is more than $800\times$ (resp. $20\times$) faster than the Core-Approx on FriendFeed (resp. DBLP), the density of its solution, with varying $\beta$, is close to (resp. better than) the density of Core-Approx's solution.

\begin{figure}
    \centering
    \includegraphics[width=0.5\textwidth]{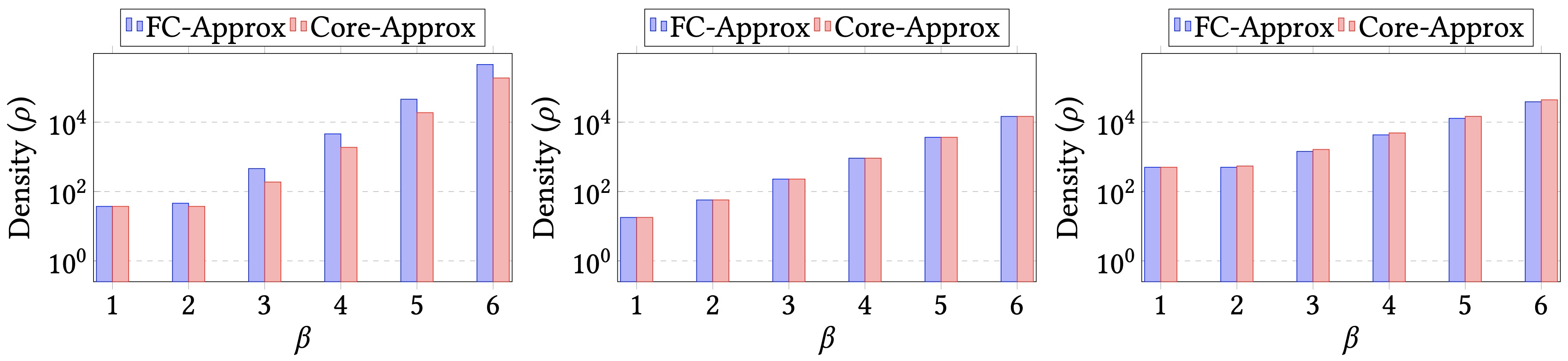}
    \vspace{-8mm}
    \caption{The value of density ($\rho$) of the output of FC-Approx and Core-Approx algorithms on DBLP (Left), Amazon (Middle), and FriendFeed (Right) datasets.}
    \label{fig:Compare_Densest}
    \vspace{-6ex}
\end{figure}

\noindent
\textbf{Case Study of DBLP.}
We present the densest subgraph extracted from the temporal DBLP network~\cite{dblp}, which is a collaboration network where each layer corresponds to  a year in 2016--2020, as an illustration of the superior quality of FC-Approx's solution and a qualitative comparison of FirmCore with \textbf{k}-core. Figure~\ref{fig:case-study} reports the algorithm's output with $\beta = 1.1$. The subgraph contains 31 nodes and 4 layers, and $\rho$ value of $45.06$. The Core-Approx algorithm finds a subgraph with 39 nodes and a $\rho$ value of $41.71$. As  discussed in $\S$~\ref{sec:Introduction}, one drawback of the previous ML cores is forcing \textit{all} nodes in the subgraph to satisfy a degree constraint, while for each node, some layers may be noisy. This subgraph corresponds to an annual report some of whose authors may change each year, and as a result, Core-Approx found~a~subgraph~with~lower~density.

We compare the robustness of ML perspective against single-layer perspective, by  collapsing all layers:  two researchers are neighbors if they collaborate in any year in 2016-2020. We find the exact densest subgraph using a max-flow based algorithm \cite{densest_first}. The densest subgraph is a $0.42$-quasi-clique with $201$ vertices and $8636$ edges. The approximate densest subgraph in the ML perspective is a $0.65$-, $0.75$-, $0.82$-, and $0.71$-quasi-clique in years $2020$, $2019$, $2018$, and $2017$, respectively. This shows that FirmCore can be used to extract hidden relationships, not seen in a~single-layer~perspective.

\begin{figure}[t]
    \centering
    \includegraphics[width=0.7\linewidth, height=0.35\linewidth]{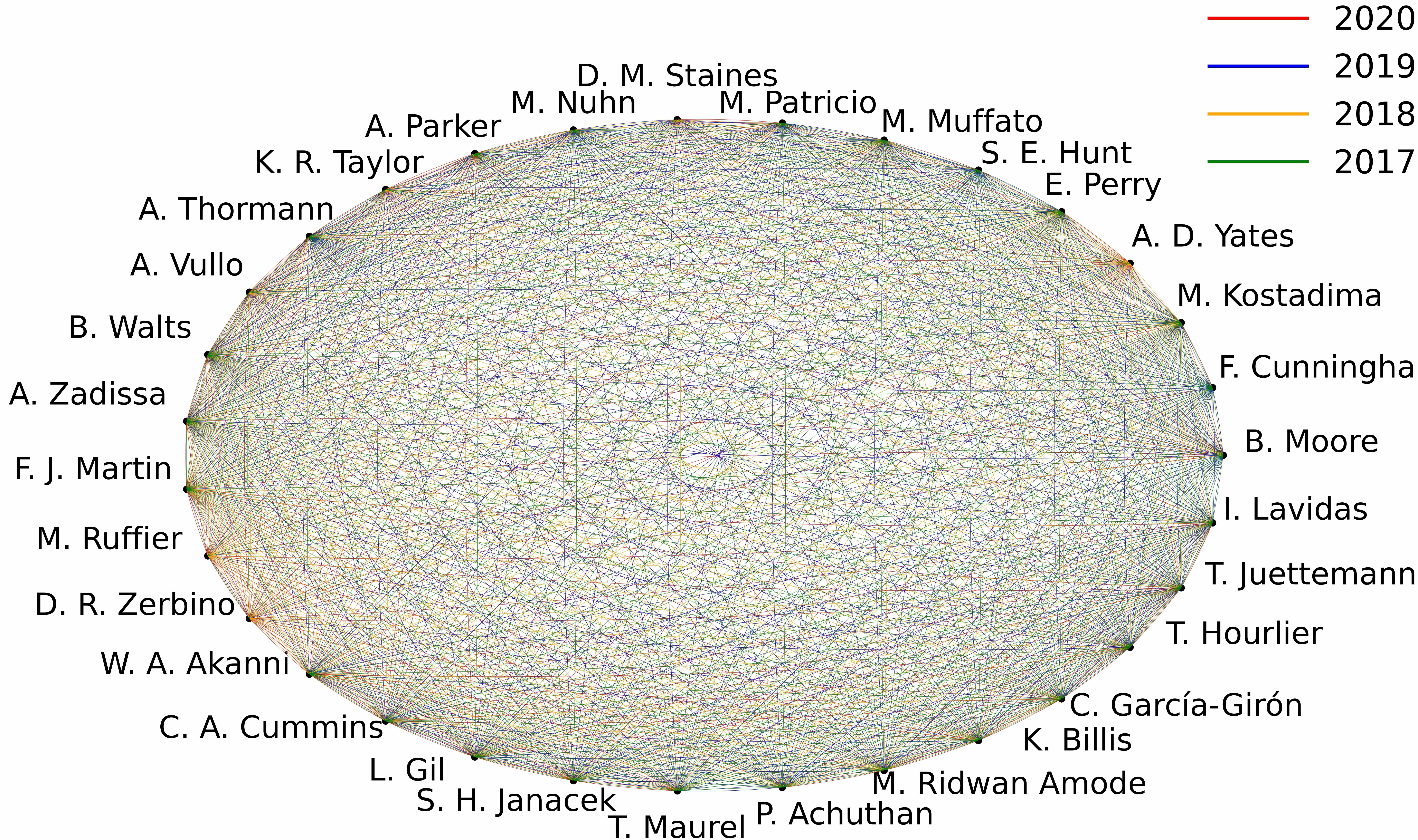}
    \vspace{-2.5ex}
    \caption{Multilayer densest subgraph extracted by Algorithm~\ref{alg:Approx1} (FC-Approx) from the DBLP dataset.}
    \label{fig:case-study}
\end{figure}

\noindent
\textbf{Case study of FAO.} We apply FDC-Approx~on~FAO dataset and our results show that the first two leading exporter~countries~of each good, which correspond to layers, are in the~found $S$~subgraph (Appendix~\ref{sec:case_FAO}).

\vspace*{-2ex} 
\section{Conclusions}
\label{sec:Conclusions}
In this work, we propose and study a novel  extended notion of core decomposition in undirected and directed ML networks, FirmCore and FirmD-Core, and establish its nice properties. Our decomposition algorithms are linear in the number of nodes and edges and quadratic in the number of layers, while the state of the art is exponential in the number of layers. 
Moreover, we extend the densest subgraph mining problem to directed ML graphs. We develop efficient approximation algorithms for this problem in both undirected and directed cases; to our knowledge, no previous polynomial-time approximation algorithms are known for these problems. We  show with detailed  experiments  that FirmCore decomposition leads to algorithms that are significantly more time and memory efficient than the start-of-the-art, with a quality that matches or exceeds.

\bibliographystyle{ACM-Reference-Format}
\bibliography{main}

\appendix


\section{Reproducibility}
\label{sec:Reproducibility}
All algorithm implementations and code for our experiments are available at \href{https://github.com/joint-em/FirmCore}{Github.com/joint-em/FirmCore}.

\section{Pseudocode of FDC-Approx}
\label{sec:FDC-Approx-appendix}
Algorithm~\ref{alg:Approx2} reports the pseudocode of the method for finding the approximate solution of densest directed subgraph problem in ML networks, discussed in Section~\ref{sec:FDC-Approx}.
\vspace{-1.5ex}
\begin{algorithm}[ht]
    \small
    \caption{FDC-Approx}
    \label{alg:Approx2}
    \begin{algorithmic}[1]
        \Require{A multilayer digraph $G = (V, E, L)$ and $\beta \in \mathbb{R}^+$}
        \Ensure{The densest $(S^+, T^+)$-induced FirmD-Core of $G$}
        \State $\Omega \leftarrow$ FirmD-CoreDecomposition(G)
        \State \Return $\arg \max_{(S, T) \in \Omega} \rho(S, T)$
    \end{algorithmic}
\end{algorithm}

\vspace{-3ex}
\section{Proofs of Propositions}
\textbf{Proposition~\ref{prop:FirmCore1}
:}
 Suppose $C_{(k, \lambda)}$ and $C'_{(k, \lambda)}$ are two distinct $(k, \lambda)$-FirmCores of $G$. By Definition 1, $C_{(k, \lambda)}$ is a maximal subgraph s.t. $ \forall v \in C_{(k, \lambda)}$ there are $\geq \lambda$ layers where the degree of $v$ is $\geq k$. Similarly, $C'_{(k, \lambda)}$ is a maximal subgraph with the same property. Then the subgraph $C_{(k, \lambda)} \cup C'_{(k,\lambda)}$ trivially satisfies the FirmCore conditions, contradicting~the~maximality~of~$C_{(k, \lambda)}$~and~$C'_{(k,\lambda)}$.~\qed 
 
 \noindent
 \textbf{Proposition~\ref{prop:FirmCore2}:} The property follows from the definition of FirmCore and this fact that in a subgraph, every node $v$ that has at least $k + 1$ neighbors, also has at least $k$ neighbors. Accordingly, in each of the (at least) $\lambda$ layers that $v$ has degree $k + 1$, it also has degree $k$. So it is in the $(k, \lambda)$-core as well.  Similarly, if vertex $v$ in at least $\lambda + 1$ layers has degree no less than $k$, in at least $\lambda$ layers it has degree at least $k$ as well. \qed

\noindent
 \textbf{Proposition~\ref{prop:quasi-clique-prop}:}
Suppose there is a node $v$ in quasi-clique, $H$, that does not in the $(\gamma (k - 1), \lambda)$-FirmCore. By definition of FirmCore, $v$ in less than $\lambda$ layers can have a degree at least $\gamma (k - 1)$. This leads to violating the definition of frequent cross graph quasi-clique since $v$ cannot be in a quasi-clique in more than $\lambda - 1$ layers.\qed

\section{Guarantee of FC-Approx}
\noindent
\textbf{Proof of Lemma~\ref{lemma2}:}
By Lemma (1), using the fact that $k^+ \geq \mu(S^*_{\text{SL}}, \ell^*)$, and $C^* = \arg \max_{C_{k,\lambda}\in \Omega} \: \rho(C_{k,\lambda})$, we have:
\begin{align*}
\rho(C^*) &\geq \rho(C^+) \geq \frac{k^+}{2 |L|} \underset{\underset{0 \leq \xi <\lambda^+}{\xi \in \mathbb{Z}}}{max} \: (\lambda^+ - \xi) (\xi + 1)^\beta \geq\frac{\psi_\beta}{2 |L|}  \mu(S^*_{\text{SL}}, \ell^*). \qed 
\end{align*}

\noindent
\textbf{Proof of Lemma~\ref{lemma3}:}
Since $S^*_{\text{SL}}$ maximizes the density in layer $\ell^*$,  removing the node with the minimum degree cannot increase its density. Accordingly,
$$\frac{|E_{\ell^*} [S^*_{\text{SL}}]|}{|S^*_{\text{SL}}|} \geq \frac{|E_{\ell^*} [S^*_{\text{SL}}]| - \mu(S^*_{\text{SL}}, \ell^*)}{|S^*_{\text{SL}}| - 1},$$
Simplifying, we get $\mu(S^*_{\text{SL}}, \ell^*) \geq \frac{|E_{\ell^*} [S^*_{\text{SL}}]|}{|S^*_{\text{SL}}|}$. Moreover, based on the definition of $\rho(.)$ and $S^*_{\text{SL}}$, we have:
$$\frac{|E_{\ell^*} [S^*_{\text{SL}}]|}{|S^*_{\text{SL}}|} |L|^\beta \geq \max_{\ell \in L} \frac{|E_\ell[S^*]|}{|S^*|} |L|^\beta \geq \max_{\hat{L} \subseteq L} \min_{\ell \in \hat{L}} \frac{|E_\ell[S^*]|}{|S^*|} |\hat{L}|^\beta = \rho(S^*)$$
\\ Hence, $\rho(S^*) \leq \mu(S^*_{\text{SL}}, \ell^*) |L|^\beta.$\qed

\vspace{1ex}
\noindent
\textbf{Proof of Theorem~\ref{theorem1}:}
Based on the Lemmas 2 and 3 we have:
\begin{align*}
\qquad \quad \: \: \: \quad \rho(C^*) \geq \frac{\psi_\beta}{2 |L|}  \mu(S^*_{\text{SL}}, \ell^*) \geq \frac{\psi_\beta}{2 |L|^{\beta + 1}} \rho(S^*). \qquad \quad \: \: \: \qed
\end{align*}
\section{Guarantee of FDC-Approx}

\setcounter{theorem}{4}

\noindent
\textbf{Proof of Lemma~\ref{lem:lemma5}:}
Similar to Lemma \ref{lemma1}, we use the pigeonhole principle. By the definition, each node $v \in S$, in at least $\lambda$ layers has at least $k$ outgoing edges to vertices in $T$. Therefore, based on the pigeonhole principle, there is a layer $\ell_1$ such that there are $\frac{\lambda |S|}{|L|}$ nodes in $S$ that have no less than $k$ outgoing edges to nodes in $T$ in layer $\ell_1$. Similarly, based on the pigeonhole principle, there is a layer $\ell_2$ (not necessarily different from $\ell_1$) such that there are $\frac{\lambda |T|}{|L|}$ vertices in $T$ that have no less than $r$ incoming edges from vertices in $S$. Again, by using pigeonhole principle and ignoring $\ell_1$, there is at least a layer $\ell'_1 \neq \ell_1$ such that there are $\frac{(\lambda - 1) |S|}{L}$ nodes in $S$ that have no less than $k$ outgoing edges to nodes in $T$ in layer $\ell'_1$. Similar result can be obtained for nodes in $T$. Similar to Lemma \ref{lemma1}, by iterating the above process  we have there are $\Tilde{L}_1, \Tilde{L}_2 \subseteq L$ such that $|\Tilde{L}_1| = |\Tilde{L}_2| = \xi+1$, and: 
\begin{align*}
\rho(S, T) &= \max_{\hat{L} \subseteq L} \min_{\ell \in \hat{L}} \frac{|E_\ell(S, T)|}{\sqrt{|S| |T|}} |\hat{L}|^\beta \geq \max_{\hat{L} \in \{ \Tilde{L}_1, \Tilde{L}_2 \}} \min_{\ell \in \hat{L}} \frac{|E_\ell(S, T)|}{\sqrt{|S| |T|}} |\hat{L}|^\beta\\
&\geq \max \{ \frac{\frac{(\lambda -  |\Tilde{L}_1|  + 1) |S|}{|L|} \times k}{\sqrt{|S| |T|}}  |\Tilde{L}_1| ^\beta, \frac{\frac{(\lambda -  |\Tilde{L}_2 | + 1) |T|}{|L|} \times r}{\sqrt{|S| |T|}}  |\Tilde{L}_2| ^\beta \} \:\:\: \\
&= \frac{(\lambda - \xi) (\xi+1)^\beta}{|L|} \max \{k \sqrt{a}, \frac{r}{\sqrt{a}}\}
\end{align*}
The right hand side of the above inequality is a function of $\xi$, so  $\rho(S, T)$ is not less than the maximum of this function exhibited by arbitrary integers $1 \leq \xi < \lambda$.\qed

\begin{lemma}
\cite{xy-core} Given a directed single-layer graph $G = (V, E)$, let $G[S^*, T^*]$ be its directed densest subgraph, and $[x^*, y^*]$ be the maximum cn-pair, we have $2 \sqrt{x^*y^*} \geq \frac{|E(S^*, T^*)|}{\sqrt{|S^*||T^*|}}$.
\end{lemma}

\setcounter{theorem}{5}
\begin{lemma}
$\frac{|E_{\ell^*} [S^*_{\text{SL}}, T^*_{\text{SL}}]|}{\sqrt{|S^*_{\text{SL}}| |T^*_{\text{SL}}|}} |L|^\beta \geq \rho(S^*, T^*)$.
\end{lemma}
\noindent
\textbf{Proof:}
By the definition of $\rho(.)$ and $(S^*_{\text{SL}}, T^*_{\text{SL}})$, we have:
\begin{align*}
 \frac{|E_{\ell^*} [S^*_{\text{SL}}, T^*_{\text{SL}}]|}{\sqrt{|S^*_{\text{SL}}| |T^*_{\text{SL}}|}} |L|^\beta &\geq \max_{\ell \in L} \frac{|E_\ell[S^*, T^*]|}{\sqrt{|S^*| |T^*|}} |L|^\beta \\
&\geq \max_{\hat{L} \subseteq L} \min_{\ell \in \hat{L}} \frac{|E_\ell[S^*, T^*]|}{\sqrt{|S^*| |T^*|}} |\hat{L}|^\beta  = \rho(S^*, T^*). \qed& &
\end{align*}

\noindent
\textbf{Proof of Theorem~\ref{theorem2}:}
Based on Lemmas~\ref{lem:lemma5}, 5, and 6, if $a = \frac{|\hat{S}|}{|\hat{T}|}$:
\begin{align*}
 \rho(S^+, T^+) \geq \rho(\hat{S}, \hat{T}) \geq \underset{\xi \in \mathbb{Z}, 0\leq \xi < \hat{\lambda}}{\max} \frac{(\hat{\lambda} - \xi) (\xi + 1)^{\beta}}{|L|} \times \max \{\hat{k} \sqrt{a}, \frac{\hat{r}}{\sqrt{a}}\}\\
 \geq \frac{\psi_\beta}{|L|}\sqrt{\hat{r} \hat{k}} \geq \frac{\psi_\beta}{|L|}\sqrt{x^* y^*} \geq \frac{\psi_\beta}{2|L|} \frac{|E_{\ell^*} [S^*_{\text{SL}}, T^*_{\text{SL}}]|}{\sqrt{|S^*_{\text{SL}}| |T^*_{\text{SL}}|}}
 \geq \frac{\psi_\beta}{2|L|^{\beta + 1}}\rho(S^*, T^*).\qed
\end{align*}

\section{Memory Usage}
\label{sec:Memory}
Figure~\ref{fig:Compare_memory} shows the result of comparing the memory usage of FirmCore with the state-of-the-art algorithms. FirmCore in all datasets needs less memory.  
\vspace{-2ex}

\begin{figure}
    \centering
    \includegraphics[width=0.47\textwidth]{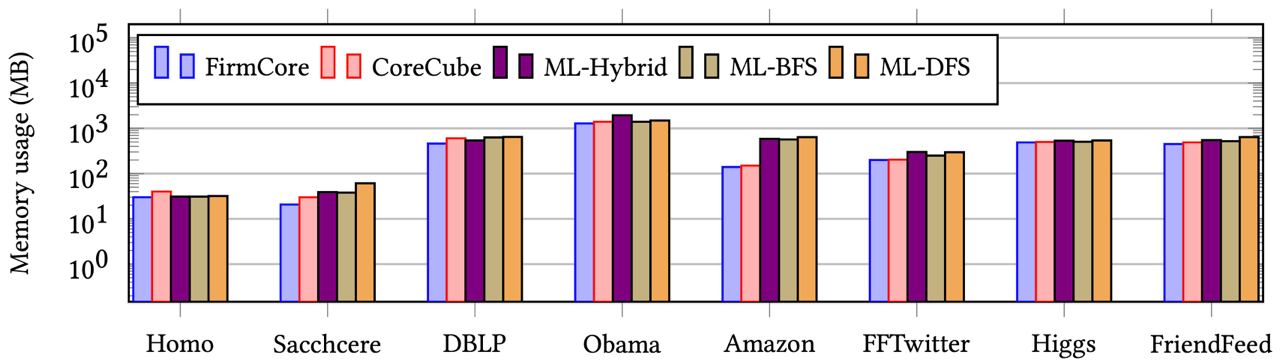}
    \vspace{-4mm}
    \caption{The memory usage of the FirmCore algorithm, ML-core algorithms, and CoreCube algorithm.}
    \label{fig:Compare_memory}
\end{figure}

\section{Density Analysis}
\label{sec:density_analysis}
While we guarantee an average degree density ($\rho$) for each FirmCore in Lemmas~\ref{lemma1} and \ref{lem:lemma5}, to show that FirmCore is able to find a large number of dense structures as per other density measures, we plot all FirmCores with number of vertices versus edge density. Figure~\ref{fig:Density} shows the plots for StackOverflow and FAO datasets. For any set $S$ of nodes, the density of the induced subgraph that we use is $\frac{|E[S]|}{|L|\cdot{|S|\choose 2}}$. Results suggest while FirmCore captures several subgraphs near clique (density $\approx 1$), it finds large sets of lower density as well. Notice that $0.15-0.2$ is a significant density for~sets~of~hundreds~of~nodes.

\begin{figure}[ht]
    \centering
    \includegraphics[width=0.39\linewidth]{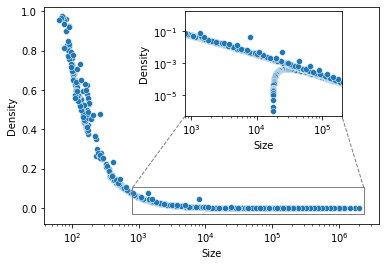}
    \includegraphics[width=0.41\linewidth]{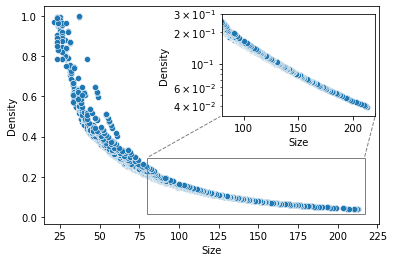}
    \vspace{-4mm}
    \caption{Edge density vs. size plots for FirmCore on FAO (Right) and StackOverflow (Left) datasets.}
    \label{fig:Density}
\end{figure}

\vspace{-4ex}
\section{Frequent Quasi-cliques}
\label{sec:frequent_quasi_clique}
As we discussed in Proposition~\ref{prop:quasi-clique-prop}, the FirmCore decomposition can speed up the extraction of frequent quasi-cliques, which is an important but NP-hard problem in ML networks. Table~\ref{tab:quasi-clique} shows the experimental results of Proposition~\ref{prop:quasi-clique-prop}'s effect on the search space of the problem. While the results are similar for other datasets, we report it on DBLP dataset because of the uniformity of the edge density across the layers. Moreover, we vary one parameter at a time keeping the other two fixed. The results suggest a significant pruning of the space by FirmCore, which in the worst case (resp. best case) prunes 87$\%$ (resp. $99.99\%$) of the search space.

\begin{table} [htpb!]
 \caption{Effect of Proposition \ref{prop:quasi-clique-prop} on the search space of the extraction of frequent cross-graph quasi-cliques on DBLP.}
 \vspace{-2ex}
\begin{center}
    \resizebox{0.43\textwidth}{!} {
\begin{tabular}{c | c |c |c |c |c | c }
 \toprule
  $\Gamma$ & $min\text{\_}sup$ & $min\text{\_}size$ & $\# solutions$ & $|V'|$ & |V| & $\%$ of pruning\\
  \midrule
   \multirow{8}{*}{$(.5, .5, .5, .5, .5, .5, .5, .5, .5, .5)$} & .5 & \multirow{3}{*}{3} & 8 & \textbf{4878} & \multirow{11}{*}{513629} & 99.05\\\cline{2-2}\cline{4-5}\cline{7-7}
   & .4 &  & 195 & \textbf{19472} &  & 96.21\\\cline{2-2}\cline{4-5}\cline{7-7}
   & .3 &  & 3394 & \textbf{66414} &  & 87.07\\\cline{2-4}\cline{4-5}\cline{7-7}
   & \multirow{8}{*}{.2} & 13 & 1 & \textbf{109} &  & 99.98\\\cline{3-5}\cline{7-7}
   &  & 11 & 8 & \textbf{217} &  & 99.96\\\cline{3-5}\cline{7-7}
   &  & 9 & 116 & \textbf{1274} &  & 99.75\\\cline{3-5}\cline{7-7}
   &  & 7 & 1292 & \textbf{5724} &  & 98.89\\\cline{3-5}\cline{7-7}
   &  & \multirow{4}{*}{8} & 121 & \textbf{1274} &   & 99.75\\
  \cline{1-1}\cline{4-5}\cline{7-7}
  $(.6, .6, .6, .6, .6, .6, .6, .6, .6, .6)$ &  &  & 18  & \textbf{217}  &  & 99.96\\\cline{1-1}\cline{4-5}\cline{7-7}
  $(.8, .8, .8, .8, .8, .8, .8, .8, .8, .8)$ &  &  & 13  & \textbf{109}  &  & 99.98\\\cline{1-1}\cline{4-5}\cline{7-7}
  $(1, 1, 1, 1, 1, 1, 1, 1, 1, 1)$ &  &  & 2 & \textbf{18} &  & 99.997\\
 \bottomrule
\end{tabular}
}
 \label{tab:quasi-clique}
\end{center}
\vspace{1ex}
\end{table}

\section{Correlation of FC Index and Top-$\lambda$ Degree}
\label{sec:corr}
In this part, we study the correlation of FirmCore index and Top$-\lambda$ degree. As an important characteristic of the FirmCore, we found that in real networks, there is a strong correlation between the $core_\lambda$ index of a node and its Top-$\lambda$ degree, which is its upper bound. In our experiment we remove nodes with Top$-\lambda$ degree equal to zero. Based on our results, the Spearman's rank correlation is more than $0.7$ ($p-$value $< 0.0001$) in all datasets. The Spearman's rank correlation for Higgs (resp. Amazon) dataset is $0.995, 0.973, 0.968,$ and $0.891$ (resp. $0.701, 0.723, 0.771,$ and $0.917$) for $\lambda = 1, 2, 3, 4$. Moreover, for Obama dataset (resp. Cannes) the Spearman's rank correlation is $0.815, 0.960,$ and $0.768$ (resp. $0.932, 0.969,$ and $0.833$) for $\lambda = 1, 2, 3$. Finally, the Spearman's correlation where $\lambda = 1, 2, 3$, are $0.997, 0.969,$ and $0.971$ (resp. $0.95, 0.96$, and $0.71$), on FriendFeed (resp. UCL).

\begin{figure}[htpb!]
    \centering
    \includegraphics[width=0.43\linewidth]{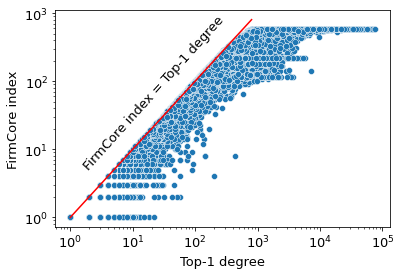}
    \includegraphics[width=0.43\linewidth]{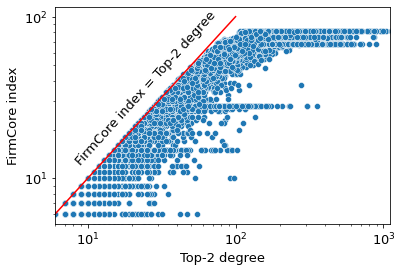}
    \vspace{-2.5ex}
    \caption{Spearman's rank between FirmCore index and Top-$\lambda$ degree in FriendFeed network.}
    \label{fig:correlation}
    \vspace{-2ex}
\end{figure}

\vspace{-2ex}

\begin{figure}[htpb!]
    \centering
    \includegraphics[width=0.835\linewidth]{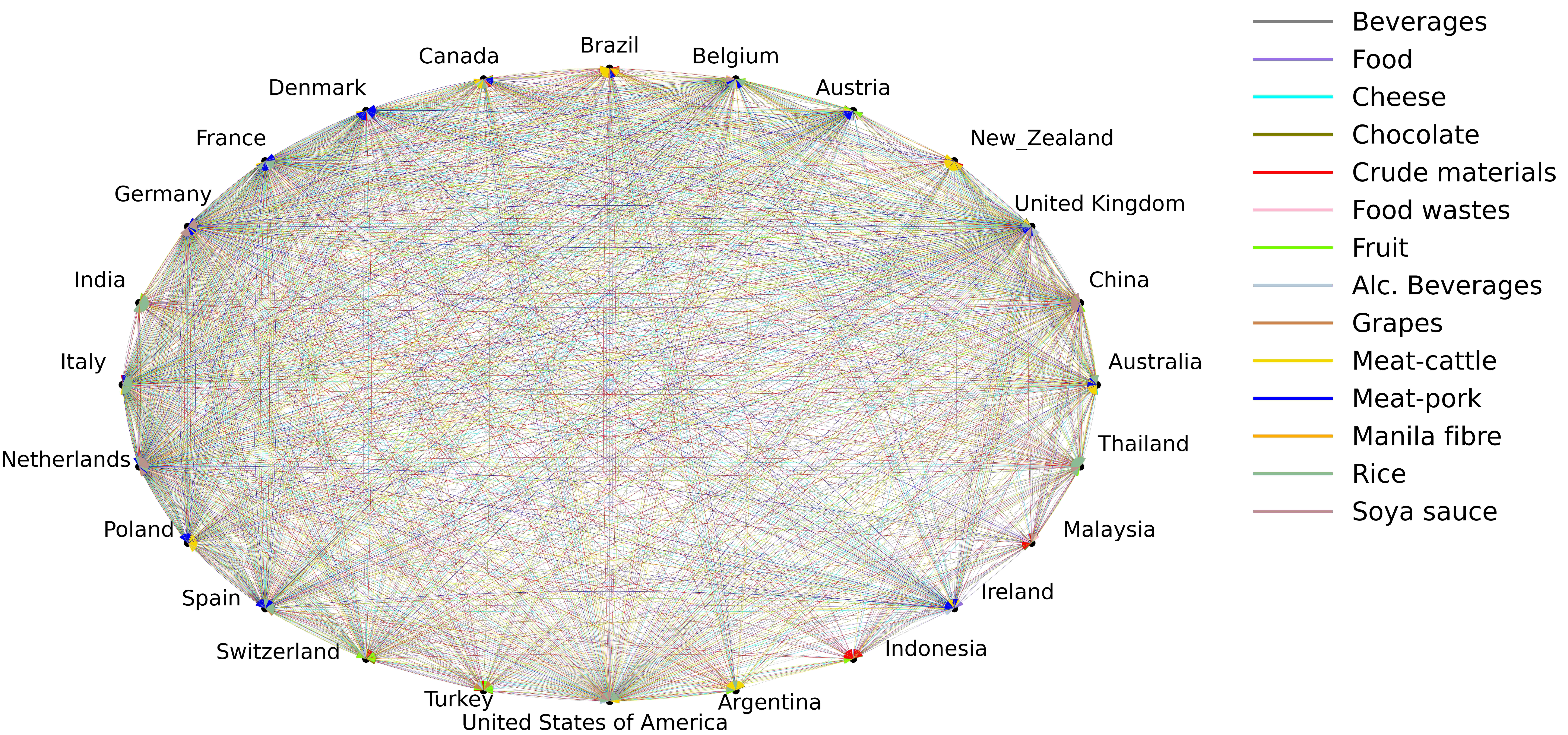}
    \vspace{-3.5mm}
    \caption{$S$ subgraph of $(S, T)-$densest subgraph extracted by FDC-Approx from the FAO dataset.}
    \label{fig:case-study-fao}
    \vspace{-2ex}
\end{figure}

\section{Case Study of FAO}
\label{sec:case_FAO}
To show the effectiveness of FDC-Approx algorithm, we report its output with $\beta = 0.5$ on the FAO dataset. Let $H = G[S, T]$ be the densest subgraph found by FDC-Approx, Figure~\ref{fig:case-study-fao} reports the subgraph $S$. Notice that $S$ and $T$ are not necessarily disjoint, and in this case $S \subseteq T$. The $(S, T)-$induced subgraph contains 24 nodes in $S$, $164$ nodes in $T$, and 14 layers, which are automatically selected by the objective function $\rho(.)$. The objective function value is $19.17$, and the minimum average degree encountered on the layer corresponding to ``Soya Sauce'' is $0.3$. By the nature of the dataset, the densest subgraph shows the countries with a large number of trades where $S$ reports countries with a large number of exports. In the output of FDC-Approx, the third (resp. seventh) layer corresponds to trades on ``Cheese'' (resp. ``Fruit''), which Germany, Netherland, and Italy (resp. Brazil, China, India) are the leading exporters\footnote{\url{https://www.statista.com}} in the world. The results for other layers are the same; in fact, our results show that the first two leading exporters countries of each good, which corresponds to layers, are in the $S$.

\end{document}